%% file: main.tex
\pdfoutput=1 
\documentclass[11pt]{article}
\input{preamble}
\title{Barriers to Collusion-resistant Transaction Fee Mechanisms}
\date{}
\author{%
    Yotam Gafni\\
    \texttt{yotam.gafni@gmail.com}\\
    Weizmann Institute%
    \and
    Aviv Yaish\\
    \texttt{aviv.yaish@mail.huji.ac.il}\\
    The Hebrew University, Jerusalem%
}
\begin{document}
\maketitle
\begin{abstract}
    To allocate transactions to blocks, cryptocurrencies use an auction-like \gls{TFM}.
    A conjecture of \citeauthor{roughgarden2021transaction} \cite{roughgarden2021transaction} asks whether there is a \gls{TFM} that is incentive compatible for both the users and the miner, and is also resistant to off-chain agreements (\glsxtrshort{OCA}s) between these parties, a collusion notion that captures the ability of users and the miner to jointly deviate for profit.
    The work of \citeauthor{chung2023foundations} \cite{chung2023foundations} tackles the problem using the different collusion resistance notion of side-channel proofness (\glsxtrshort{SCP}), and shows an impossibility given this notion.
    We show that \glsxtrshort{OCA}-proofness and \glsxtrshort{SCP} are different, with \glsxtrshort{SCP} being strictly stronger.
    We then fully characterize the intersection of deterministic \gls{DSIC} and \glsxtrshort{OCA}-proof mechanisms, as well as deterministic \glsxtrshort{MMIC} and \glsxtrshort{OCA}-proof ones, and use this characterization to show that only the trivial mechanism is \gls{DSIC}, \gls{MMIC} and \glsxtrshort{OCA}-proof.
    We also show that a randomized mechanism can be at most $0.842$-efficient in the worst case, and that the impossibility of a non-trivial \gls{DSIC}, \gls{MMIC} and \glsxtrshort{OCA}-proof extends to a couple of natural classes of randomized mechanisms.
\end{abstract}


\section{Introduction}
Cryptocurrencies such as Bitcoin utilize a decentralized mechanism wherein entities called \emph{miners} process user transactions in batches called \emph{blocks}.
In times of congestion, the prompt processing of a transaction may be incentivized by attaching a \emph{fee} to it, which can collected by miners upon including the transaction in a block.
A cryptocurrency's \gls{TFM} determines how much fees are charged from transactions and transferred to miners as revenue, with \glspl{TFM} commonly modeled as auctions.
A key difference between \glspl{TFM} and traditional auctions is the ease with which miners and users may collude via smart contracts to increase their profits. The possibility of designing such \textit{collusion-resistant} \glspl{TFM} is the main focus of our work. 

The seminal work of Roughgarden \cite{roughgarden2021transaction} includes collusion resistance, referred to as \gls{OCA}-proofness, as the centerpiece of ``good'' \glspl{TFM}, alongside other properties such as being incentive compatible for both users and miners (\gls{DSIC} and \gls{MMIC}, respectively).
Given this desiderata, \citeauthor{roughgarden2021transaction} poses an important open question: \emph{can we design good \glspl{TFM} that satisfy all desired properties?}

This was followed by Chung \& Shi \cite{chung2023foundations}, who discuss a collusion resistance notion they call $c$-\gls{SCP}, which means that no coalition of the miner and $c$ users can benefit from colluding.
They show that satisfying \gls{DSIC} and $1$-\gls{SCP} is only possible for the trivial mechanism that never allocates any item, implying that good \glspl{TFM} produce $0$ miner revenue.
This is exacerbated by similar results for \glspl{TFM} that are allowed to use strong cryptographic primitives, even when considering relaxed Bayesian and approximate versions of the \gls{TFM} deisderata \cite{shi2023what,chen2023bayesian,maximizingMinerRevenue}.

As we show, \citeauthor{roughgarden2021transaction}'s \cite{roughgarden2021transaction} open question is not directly answered by the results for \gls{SCP} \cite{chung2023foundations}, as \gls{OCA}-proofness and \gls{SCP} are not equivalent.
To answer the question about \gls{OCA}-proofness, we offer the following results:
\begin{itemize}
\item We show the existence of \gls{OCA}-proof mechanisms that are not \gls{SCP} in \cref{clm:OcaVsScp}, and then show in \cref{clm:ScpImpliesOca} that \gls{SCP} is strictly stronger as it implies \gls{OCA}-proofness.

\item In \cref{lem:SingleBidderRevenue}, we show that any DSIC+MMIC+OCA-proof mechanism must have $0$ miner revenue in the single bidder case.

\item For the general case, we answer \citeauthor{roughgarden2021transaction}'s \cite{roughgarden2021transaction} open question in \cref{thm:det_impossibility} by showing that the trivial \gls{TFM} is the only possible determinstic mechanism satisfying DSIC, MMIC, and $1$-OCA-proofness\footnote{Notice that we introduce a $c$ quantifier for the \gls{OCA} collusion notion, similarly to the quantifier for \gls{SCP}.}. We do so by precisely characterizing all DSIC+$1$-OCA mechanisms, and all MMIC+$1$-OCA-proof mechanisms, which is of independent interest. 

\item For randomized mechanisms, we show that the trivial mechanism is the only DSIC + MMIC + OCA-proof, \textit{scale-invariant} mechanism (\cref{thm:scale_invariant}).
Scale invariance is a natural property, defined by requiring that scaling all bids by the same factor leaves the allocation unchanged.
It holds for notable mechanisms such as the first-price, second-price, and the all-pay auctions.
It does not hold, for example, for the second-price auction with reserve price $r > 0$.

\item We show that no randomized DSIC+MMIC+OCA-proof mechanism is efficient in \cref{thm:ctpa}, and we bound away the efficiency approximation in \cref{thm:allocation_bound} and Corollary~\ref{corr:efficiency_bound}, showing that it can achieve at most $0.842$ in the worst case. 

\end{itemize}

\subsection{Technical Overview}
We start our work by fully understanding the single bidder case.
In Lemma~\ref{lem:SingleBidderRevenue}, we observe that due to the analogy in this case between the bidder's utility (w.r.t. payments) and the joint utility of the bidder and miner (w.r.t. to the amount of fees that are burnt), it must hold that all payments are burnt, meaning that miner revenue must equal 0.
As we show in Example~\ref{ex:second_price}, this analogy does not hold once more bidders are involved.
However, this important observation paves the way for our general result on the deterministic multi-bidder case.
Thus, in Theorem~\ref{thm:zero_revenue} we show that the miner cannot have any revenue in the general case, as otherwise this would allow the miner to achieve positive revenue in the single bidder case by introducing fake bids
.

Going beyond this result, we fully characterize the OCA-proof single bidder case and find it takes the form of a ``posted burn’’ (in analogy to the ``posted price’’ characterization of the single bidder DSIC case).
I.e., there is some reserve ``burn'' $r$, so that the price the bidder must pay is equal to or above it, with $r$ completely burnt.
The mechanism has freedom in choosing the payment as long as it is larger than $r$, as to allow burning the reserve price in full.
Then, we extend this structure to any number of bidders, again due to the OCA-proofness property.
That is since the highest bidder and the miner always have the option to omit all other bids and give the item (subject to the ``posted burn’’ $r$) to the highest bidder, thus maximizing their joint utility.
We therefore derive a general structure for OCA-proof mechanisms as allocating to the highest bidder with a ``posted burn’’ $r$. 
This generalizes a known result in the TFM literature: Both the ``tipless mechanism’’ and EIP-1559 are OCA-proof \cite{roughgarden2021transaction} when the burn rate is high enough, where both have different payments.
Our result emphasizes that, in fact, there can be a wide variety of mechanisms of this sort. 

This characterization of OCA-proof mechanisms, where the payment is the only degree of freedom left for the mechanism besides deciding the ``posted burn’’ $r$, is then what drives our impossibility result for the deterministic case.
This is because DSIC and MMIC payment rules differ conceptually: DSIC rules must follow the Myerson critical bid form (Fact~\ref{fct:MyersonLemma}), and so DSIC+OCA-proof mechanisms follow the form of a second-price auction with a reserve price $r$ that is burned (Theorem~\ref{thm:dsic_1oca}), while MMIC payment rules must be robust to the miner artificially increasing them through fake bids, and so must only depend on the winning bid.
Strictly speaking, the payment does not have to be the winning bid itself (i.e., first-price with burned reserves), but rather any function that increases with the winning bid and respects the ``posted burn’’ $r$.
This can be considered as a form of ``generalized first-price’’ (Theorem~\ref{thm:mmic_1oca}).
Our characterizations of all DSIC+OCA-proof and MMIC+OCA-proof mechanisms, have just a single mechanism at their intersection: the trivial mechanism, thus resulting in our $0$ revenue impossibility result for DSIC+MMIC+OCA-proof mechanisms (Theorem~\ref{thm:det_impossibility}).

We then turn our attention to randomized mechanisms.
We first consider the natural class of \textit{scale invariant} mechanisms.
For such mechanisms, any homogeneous transformation of the bids preserves the allocation outcome.
These mechanisms are natural, as intuitively, the relative values of the different bidders should be the decisive factor when allocating bids, and not the ``measurement unit’’.
Indeed, many commonplace auctions are scale invariant, e.g., the first-price auction, the second-price auction, and the all-pay auction.
The crucial technical property of scale-invariant OCA-proof auctions is that they cannot burn any fees.
This is because given a non-zero burn rule, bidders can coordinate to all scale down their bids, to the point that this burn exceeds what is allowed by basic auction properties such as individual rationality and burn balance. Thus, the burnt amount must be lowered, in contradiction to OCA-proofness.
Together with our understanding of the single bidder case as having $0$ revenue, this implies that the item must be either always or never given to the single bidder (Lemma~\ref{lem:always_or_never_allocated}).
Having the item fully allocated to the bidder with no burn in the single bidder case is very lucrative for the joint utility of the bidder and the miner, as they can achieve optimal utility.
With OCA-proofness, this degenerates the multi-bidder case to be forced to give the item to the highest bidder with full probability, which brings us back to the deterministic case, and thus to an impossibility (Theorem~\ref{thm:scale_invariant}). 

Lastly, we show efficiency approximation hardness results for general randomized mechanisms (i.e., not necessarily scale-invariant).
Key to these results is Claim~\ref{clm:payment_burn_bound}, which shows a lower bound on the possible aggregate burn for the two bidder case.
While in the deterministic case we were able to show zero revenue, the miner has harder time introducing fake bids in the randomized case.
That is because a fake bid might also win with some probability, and thus resulting in having some of its payment burnt.
Still, while not implying zero revenue for the miner, this at least implies that its payment from the real bid must be bounded by the aggregate burn.
Since the joint utility expression depends on the aggregate burn, this serves to 
bound the two-bidder allocation rule with the single bidder allocation rule (Lemma~\ref{lem:two_bidders_cond}), by first tying together the allocation and payment through Myerson’s lemma, and then the payment and burn through Claim~\ref{clm:payment_burn_bound}. 

Using the way Lemma~\ref{lem:two_bidders_cond} ties the two-bidder case together with the single-bidder case, we are then able to infer two important results, both limiting the efficiency approximation that mechanisms may achieve.
First, in Theorem~\ref{thm:allocation_bound}, we upper bound the maximal probability with which the item is allocated in the single bidder case.
No matter how high the single bid is, the item cannot be allocated with probability higher than ~0.914.
Then, in Lemma~\ref{lem:two_bidder_bounds}, we tie together the allocation probability in the two bidder case with the utility in the single bidder case.
We use this in Corollary~\ref{corr:efficiency_bound} to upper bound the utility in single bidder case (depending on the bidder’s value).
Notice that because of our scale-invariance result, we know that the mechanism must behave differently given different bidder’s values.
Corollary~\ref{corr:efficiency_bound} shows that for some bidders' values it must behave quite badly, allowing for efficiency approximation, i.e., the ratio between the joint utility and the optimal joint utility where the item is fully allocated without any burn, is at most 0.842. 

\subsection{Related Work}

\subsubsection*{Transaction Fee Mechanisms}
The literature on applying auction theory to analyze blockchain transaction fee mechanisms has its early start in \citeauthor{lavi2019redesigning} \cite{lavi2019redesigning}, who put forward the monopolistic auction as a good TFM candidate that balances incentive-compatibility for users and the miner.
\citeauthor{yao2018incentive} \cite{yao2018incentive} answers two open conjectures of \cite{lavi2019redesigning} by providing a proof that the monopolistic auction is indeed asymptotically \gls{DSIC} and that its revenue dominates the random sampling optimal price (RSOP) auction of \citeauthor{competitiveAuctions} \cite{competitiveAuctions}.
\citeauthor{roughgarden2021transaction} \cite{roughgarden2021transaction,roughgarden2020transaction} introduces the OCA collusion-resistance notion and characterizes different aspects of the TFM space. He also provides detailed analysis of EIP-1559 and other popular auction formats.
\citeauthor{chung2023foundations} \cite{chung2023foundations} introduce the SCP collusion-resistance notion and show related properties, most importantly the impossibility of DSIC+$1$-SCP mechanisms. 

\citeauthor{maximizingMinerRevenue} \cite{maximizingMinerRevenue} further explore the design of \glspl{TFM} when assuming honesty of most or some of the bidders.
They show a very general impossibility of achieving non-negligible miner revenue, even when altogether using cryptographic primitives, Bayesian assumptions, approximate incentive-compatibility, and infinite block size, and match it with mechanisms that achieve this miner revenue.
\citeauthor{shi2023what} \cite{shi2023what} present the cryptographic MPC-assisted \gls{TFM} model and show we can only achieve $\theta(\epsilon)$ miner revenue given $\epsilon$-approximate notions of incentive-compatibility and \gls{SCP}.
\citeauthor{gafni2023greedy} \cite{gafni2023greedy} show we can achieve a \gls{OCA}-proof, \gls{MMIC} and Bayesian incentive-compatible (BIC) mechanism given i.i.d. bidders. \citeauthor{chen2023bayesian} \cite{chen2023bayesian} show a BIC and \gls{SCP} mechanism.
While the focus of these works is in finding what relaxing assumptions and tools may allow feasibility of collusion-resistant mechanisms, we shine a light on matters still unresolved in the plain model, and focus on fully characterizing it.

Other works focus on reducing fees for bidders \cite{basu2023stablefees,damle2024designing}, and some empirical works analyze the systemic effects of introducing EIP-1559 in Ethereuem \cite{liu2022empirical,reijsbergen2021transaction}. 

\subsubsection*{Extending the valuation framework}
\citeauthor{bahrani2023transaction} \cite{bahrani2023transaction} consider a setting where a miner has preferences over which transactions are included in the block, aside from it wishing to receive high revenue in fees.
They show an impossibility for DSIC + MMIC mechanisms in this more general setting.
\citeauthor{nourbakhsh2024transaction} \cite{nourbakhsh2024transaction} consider a setting where \textit{users} have preferences over the \textit{ordering} of transactions in a block, and show that there is no deterministic DSIC + OCA-proof in this general order-sensitive setting.
\citeauthor{bahrani2023when} \cite{bahrani2023when} consider single-item auctions where each ``bidding agent'' is in fact a decentralized autonomous organization (DAO). The agents treat the item as a public excludable good, and distribute the payment for it among them. The authors show that it is then impossible to achieve an efficient two-level mechanism, and show a tight $\theta(\log n)$ bound on the efficiency that can be achieved, where $n$ is the largest number of bidders in a DAO.
\citeauthor{milionis2022framework} \cite{milionis2022framework} study single-item NFT auctions. Although different from TFMs, they analyze them through similar lens, and find that it is impossible to achieve incentive-compatibility together with collusion-resistance. 

\subsubsection*{The Dynamic Setting}
While we, and the previously mentioned papers, focus on the setting of a \textit{single block}, it is interesting to consider the dynamic setting, as the blockchain TFM operates repeatedly within short periods of time. The dynamic setting is not yet well-consolidated, and many different approaches exist to its modeling and operation. Different works study dynamic posted pricing \cite{ferreira2021dynamic}, dynamic manipulation on EIP-1559 \cite{azouvi2023base}, analysis of the effects of dynamics on price and efficiency \cite{kiayias2023tiered,nisan2023serial,leonardos2023optimality}, and the implications of having farsighted users and miners \cite{gafni2023greedy,hougaard2023farsighted}. 

\subsubsection*{Random vs. Deterministic Mechanisms}
Generally, in the mechanism design literature, randomized mechanisms can achieve better results than deterministic ones. 
For example, \citeauthor{randomPower} \cite{randomPower} show that randomized mechanisms can arbitrarily approximate the welfare of multi-unit auctions, while  \citeauthor{multiUnitApproxTwo} \cite{multiUnitApproxTwo} show that deterministic mechanisms may achieve at most an approximation of $2$.
\citeauthor{nisanRonen} \cite{nisanRonen} show a randomized/deterministic gap for the problem of task scheduling. In the context of \gls{TFM}s, 
\citeauthor{chung2023foundations} \cite{chung2023foundations} consider a notion where the miner is punished for introducing fake transactions (since they may need to end up paying for the fake transactions in future blocks).
Given this notion, they show a randomized mechanism (the burning second-price auction) that achieves good properties, while they show that no deterministic mechanism can achieve it. 

Despite the game-theoretic advantages of randomized mechanisms, ``good'' sources of randomness, also known as randomness beacons, are hard to come by in the blockchain setting \cite{bunz2017proofs,choi2023sok}, making deterministic mechanisms simpler to implement.
Ideal randomness beacons should periodically generate random values such that no actor should be able to predict or bias future random values.
Both properties may be vital for random mechanisms: For example, predictable beacons could allow users to account for future randomness when creating transactions so to increase their allocation probability, while miners may manipulate beacons susceptible to bias in order to increase their profits.
In particular, Ethereum's RANDAO is known to be vulnerable to biasing manipulations that allow even small agents (i.e., proposers with a small fraction of the stake) to risklessly increase their expected number of blocks to some small extent \cite{alturki2020statistical,edgington2023upgrading}.
Incorporating randomness into the transaction allocation process can increase the profitability of such manipulations, where several randomness beacon designs are known to be secure only when the potential profits from manipulations are limited \cite{choi2023sok}.
This is not only a theoretical threat as miners have been observed to manipulate cryptocurrency mechanisms in the wild \cite{yaish2023uncle}, and the potential profits from manipulations can be large \cite{zhou2023sok}.

While we only show impossibility results for both deterministic and randomized mechanisms, our results for deterministic mechanisms are more conclusive, and it remains possible that a randomized mechanism may achieve the \gls{TFM} incentive-compatibility and \gls{OCA}-proof desiderata.  

\subsubsection*{The Economic Literature on Collusion}
There is vast economic literature discussing collusion, especially concerning industrial cartels.
The focus of that literature is mostly on \textit{user-user} collusion\footnote{See, e.g., \cite{greenLaffont,deckelbaumMicali,biddingClubs,goldbergHartline}}, i.e., bidders cooperating without the auctioneer. This is in contrast to  \textit{miner-user} collusion, where the coalition necessarily includes the miner, as in the OCA and SCP notions of \cite{roughgarden2021transaction,chung2023foundations}.
This focus on user-user collusion can perhaps be explained by US government auctions being the prime example of this auction setting.
Early on, the seminal work of \citeauthor{Stigler1964} \cite{Stigler1964} identified many caveats that stand in the way of successful collusion. 

Different works focus on particular types of collusion. The seminal work of \citeauthor{ausubel2002ascending} \cite{ausubel2002ascending} studies \textit{core-selecting} auctions. This captures the fact that in combinatorial auctions with complements (rather than substitutes), VCG may charge low payments, such that some of the losing bidders may be willing to pay. Thus, for an outcome to be \textit{stable}, it must be in the core of the underlying coalitional game. \citeauthor{goereeLien2016} \cite{goereeLien2016} show an impossibility of truthful efficient core-selecting auctions, which is somewhat analogous to the DSIC + 1-SCP impossibility \cite{chung2023foundations}. 
Another prominent notion of collusion is that of \textit{bribing} \cite{fox2023censorship,esoSchummer2000,rachmilevitchFirstPrice,rachmilevitchSecondPrice}. 

\citeauthor{optimalCollusionProof} \cite{optimalCollusionProof} show how to construct optimal revenue collusion-proof single-item auctions, i.e., how to implement the Myerson auction allocation rule in a way that is collusion-resistant. Notice, however, that they assume knowledge of the identity of the colluding bidders. An interesting emphasis in the work is that beliefs and strategies are updated in case a possible member \textit{declines} participating in a bidding ring. The bidder is then aware of the existence of the bidding ring, and may also be punished by the bidders upon its rejection of participation.\footnote{\citeauthor{rachmilevitchFirstPrice} \cite{rachmilevitchFirstPrice} also emphasizes, in the case of bribing in first-price auctions, how rejecting a bribe influences the beliefs and behavior of bidders subsequently.} \citeauthor{valiant2007collusion} \cite{valiant2007collusion} provide a characterization of what revenue benchmark can be achieved for general collusion-proof combinatorial auctions, and provide a mechanism that guarantees a revenue which is a logarithmic approximation of the welfare of all but the ``highest'' bidder, in a sense they define. 

\citeauthor{mcAfeeMcMillan} \cite{mcAfeeMcMillan} provide a characterization of the optimal collusion against second-price auctions with and without transfers between the colluders. Importantly, they emphasize the bargaining aspect of the collusion, where different bidders may prefer different collusion agreements. \citeauthor{honorThieves} \cite{honorThieves} studies this problem as a coalitional game, and shows that the Shapley value of the bidders (w.r.t. the collusion) is in the core of the game.

\section{Model and Preliminaries}
\label{sec:Model}
We follow the standard model used by the literature for \gls{TFM}s \cite{lavi2019redesigning,roughgarden2021transaction,roughgarden2020transaction,chung2023foundations,shi2023what}. The main technical assumption we make, which we use to simplify the discussion, is that the auction is a \textit{single-item} auction, i.e., that the block size is $1$. We interchangeably use the terms \textit{auction}, \textit{mechanism}, or \gls{TFM} to refer to the object of our discussion, depending on what is most appropriate in the context. 
Some missing proofs are in Appendix~\ref{app:missing_proofs}. 

\subsection{Transaction Fee Mechanisms}
A \gls{TFM} $\auction$ consists of an \emph{allocation rule} $\alloc^\auction$, a \emph{payment rule} $\pay^\auction$ and a \emph{burning rule} $\burn^\auction$.
For brevity, when it is clear from context, we omit the $\auction$ identifier from the notations.

\subsubsection*{Allocation rule}
An allocation rule defines the mechanism's allocation of transactions to the upcoming block, as intended by the mechanism's designers.
\begin{definition}[Allocation rule]
    \label{def:AllocFunc}
    A \textit{deterministic} allocation rule $\alloc_n^\auction:R_+^n \rightarrow \{0,1\}^n$ of auction $\auction$ defines which transaction should be included in the upcoming block by $\auction$. I.e., we require that $\forall \mathbf{b} \in R_+^n: \sum_i [\alloc_n^\auction(\mathbf{b})]_i \leq 1$.
    More generally, an allocation rule may be \textit{randomized}, i.e., $\alloc_n^\auction:R_+^n \rightarrow \{\Delta(0,1)\}^n$, with the same feasibility restriction over the sum of allocation probabilities.
\end{definition}
\begin{remark}
    In our analysis, we omit the identifier $n$ and write $\alloc^\auction$.
    We assume that the correct $\alloc_n^\auction$ is used, based on the length $n$ of the vector of bids $\mathbf{b}$. We also use the following two auxiliary notations: 
    \begin{enumerate}
    \item We use the notation $\tilde{\alloc}^\auction:R_+^n\rightarrow [n] \cup \{\emptyset\}$ for deterministic mechanisms to specify the unique bidder who receives the item, or $\emptyset$ if the item remains un-allocated.
    I.e., if $\tilde{\alloc}^\auction(\mathbf{b}) = i$, then:
    \begin{equation*}
        \alloc^\auction(\mathbf{b}) = \{0, \ldots, 0, \underbrace{1}_{i\text{-th index}}, 0, \ldots, 0\}.
    \end{equation*}
        
    \item Instead of writing $[\alloc^\auction(\mathbf{b})]_i$ to specify the allocation probability for bidder $i$ under $\auction$, we use the notation $\alloc^\auction(b_i, b_{-i})$, where $b_i$ is bidder $i$'s bid, and $b_{-i}$ are the bids of all other bidders.
    While still emphasizing that bidder $i$ is the bidder we are interested in, this notation is better suited for varying $b_i$, while fixing $b_{-i}$.
    With this notation, we write the feasibility constraint of \cref{def:AllocFunc} as:
    \begin{equation}
    \label{eq:feasibility}
    \sum_i \alloc^\auction(b_i, b_{-i}) \leq 1
    \end{equation}
    \end{enumerate}
\end{remark}

\subsubsection*{Payment \& burn}
The payment and burn rules respectively define the amount of fees paid by each transaction when included in a block, and how much of each payment is ``burnt'' and taken out of circulation instead of being given to the miner, with both rules enforced by the mechanism.
\begin{definition}[Payment rule]
    The payment rule $\mathbf{\pay^\auction}:R_+^n\rightarrow R_+^n$ in auction $\auction$ receives bids $b_1, \ldots, b_n$, and outputs a payment vector with a payment for each bidder $i$.
\end{definition}
\begin{definition}[Burning rule]
    The burning rule $\mathbf{\burn^{\auction}}:R_+^n\rightarrow R_+^n$ in auction $\auction$ receives bids $b_1, \ldots, b_n$, and and outputs a burn vector with a burn for each bidder $i$, meaning the amount of funds out of $b_i$'s payment that is taken out of circulation.
\end{definition}

\begin{remark}
The definitions of the payment and burning rules can be extended to the randomized case, similarly to our extension of allocation functions.
For our purposes, it suffices to consider the \textit{expected} payment or burn of each bidder under randomized rules.
We also extend the notation in the two ways we did for allocation functions.
For $x \in \{\pay^\auction, \burn^\auction\}$, we write $x(b_i, b_{-i})$ instead of $[x(\mathbf{b})]_i$, and in the deterministic case, we use $\tilde{x}:R_+^n \rightarrow R$ to denote the payment or burn of the \textit{winner}, i.e., if $\tilde{\alloc}^\auction(\mathbf{b}) = i$, then: $x(\mathbf{b}) = \{0, \ldots, 0, \underbrace{\tilde{x}(\mathbf{b})}_{i\text{-th index}}, 0, \ldots, 0\}$.
\end{remark}


\subsubsection*{Basic definitions}
We proceed with several definitions that are used throughout our analysis.

\begin{definition}[Agent Utilities]
    \label{def:utilities}
    Let auction $\auction = (\alloc, \pay, \burn)$. Consider $n$ bidders with true valuations $\mathbf{v} \in R_+^n$ and $n'$ bids 
    $\mathbf{b} \in R^{n'}$. 
    We assume that $\mathbf{b}$ is indexed the same as $\mathbf{v}$, meaning that the first $n$ bids in $\mathbf{b}$ correspond to the bids in $\mathbf{v}$ (possibly as $0$ if omitted), and so w.l.o.g. we consider $n' \geq n$. We consider the $n' - n$ bidders that are in $\mathbf{b}$ but not in $\mathbf{v}$ as the miner's ``fake bidders'', and denote them by $\mathbf{b}_F$. 
    Then, the utilities of the various agents are defined as follows.
    \begin{itemize}
    \item  \textit{Bidder Utility:}
    \begin{equation}
    \label{eq:bidder_util}
        u_i(\fee_i, \mathbf{\fee_{-i}} ; v_i)
        \define
            v_i \cdot \alloc(b_i, \mathbf{b}_{-i}) - \pay\left(\fee_i,\mathbf{\fee_{-i}}\right).
      \end{equation}

    \item \textit{Miner Utility:}
    \begin{equation}
    \label{eq:miner_util}
    \begin{split}
    u_{miner} (\mathbf{\fee} ; \mathbf{v})
    &
    \define
    \sum_{i=1}^{n'} \left(\pay(b_i, \mathbf{b}_{-i}) - \burn(b_i, \mathbf{b}_{-i}) \right) - \sum_{i=n+1}^{n'} u_i(\fee_i, \mathbf{\fee_{-i}} ; 0)
    \\&
    =
    \sum_{i=1}^n \left(\pay(b_i, \mathbf{b}_{-i}) - \burn(b_i, \mathbf{b}_{-i}) \right) - \sum_{i=n+1}^{n'} \burn(b_i, \mathbf{b}_{-i}).
    \end{split}
    \end{equation}

    \item
    \textit{Joint Utility:}
    \begin{equation}
    \label{eq:joint_util}
    \begin{split}
    u_{joint}(\mathbf{b} ; \mathbf{v})
    &
    \define
    u_{miner}(\mathbf{b} ; \mathbf{v}) + \sum_{i=1}^n u_i(b_i, \mathbf{b}_{-i} ; v_i)
    \\&
    =
    \sum_{i=1}^n \left( v_i \alloc(b_i, \mathbf{b}_{-i}) - \burn(b_i, \mathbf{b}_{-i})\right) - \sum_{i=n+1}^{n'} \burn(b_i, \mathbf{b}_{-i}).
    \end{split}
    \end{equation}
    \end{itemize}
\end{definition}

\begin{definition}[Basic Auction Properties]
We say an auction is:
\begin{itemize}
    \item \textit{Individually Rational}, if the expected payment by an agent does not exceed its bid if it is allocated, and $0$ if it is not.
    Formally, for any $\mathbf{b} \in R_+^n$ and for each bidder $i$, we have
    \begin{equation}
    \label{eq:IR_util}
    u_i(b_i, \mathbf{b}_{-i} ; b_i) \geq 0.
    \end{equation}
    If we develop this condition given the explicit form of bidder utility in \cref{eq:bidder_util}, we get:
    \begin{equation}
    \label{eq:IR}
    \alloc(b_i, \mathbf{b}_{-i}) \cdot b_i \geq \pay(b_i, \mathbf{b}_{-i}).
    \end{equation}

    \item \textit{Burn-Balanced}, if the expected burn is non-negative and does not exceed bidder payments.
    Formally, for any $\mathbf{b} \in R_+^n$ and for each bidder $i$, we have
    \begin{equation}
    \label{eq:BB}
    \pay(b_i, \mathbf{b}_{-i}) \geq \burn(b_i, \mathbf{b}_{-i}) \geq 0,
    \end{equation}

    \item \textit{Anonymous}, if the allocation, payment, and burn rules are agnostic to permutations over the bidders.
    Formally, for any vector of bids $\mathbf{b} \in R_+^n$, permutation $\pi \in S_n$, and rule $x \in \{\alloc, \pay, \burn\}$, we have
    $x(\pi(\mathbf{b})) = \pi(x(\mathbf{b}))$.
    For anonymous auctions, we treat bids as an unordered set.
\end{itemize}
\end{definition}

Throughout the text, we assume the basic auction properties hold. The only exception to it is in Appendix~\ref{app:nonAnon} where we explore removing the anonymity assumption. 

We follow with \cref{clm:simple_joint_bound}, which is a useful ``helper'' result.
We provide the proof in \cref{app:missing_proofs}.
\begin{restatable}[]{claim}{SimpleJointBound}
\label{clm:simple_joint_bound}
In a single-item auction, for any $\mathbf{v}, \mathbf{b}$, the joint utility is at most equal to the highest valuation: $u_{joint}(\mathbf{b} ; \mathbf{v}) \leq \max \mathbf{v}$. 
\end{restatable}

\subsection{The TFM Desiderata}
\label{sec:Desiderata}
We now go over the literature's standard desiderata for ``good'' \glspl{TFM}, which includes the \gls{DSIC} and \gls{MMIC} properties.
These properties are usually augmented by some collusion resistance notion, i.e., either \gls{OCA}-proofness or \gls{SCP}.

\begin{definition}[\Glsxtrfull{DSIC}]
    \label{def:DSIC}
    An auction is \gls{DSIC} if it is always weakly better for bidders to declare their true values.
    Thus, for any bidder $i$ with true value $v_i$, for any bid $b_i$, and for any bids $\mathbf{b_{-i}}$ by all other bidders, we have
    $
    u_i(v_i, \mathbf{b_{-i}} ; v_i) \geq u_i(b_i, \mathbf{b_{-i}} ; v_i),
    $
    and there is such $\mathbf{b_{-i}}$ so that the inequality is strict. 
\end{definition}

\begin{definition}[\Glsxtrfull{MMIC}]
    A \gls{TFM} is \gls{MMIC} if a miner cannot strictly increase its revenue by deviating from the intended allocation rule, or by introducing ``fake'' miner-created bids.
    Formally, for any miner strategy $\mu$, and any valuation $\mathbf{v}$,
    \begin{equation}
    \label{eq:mmic_cond}
    u_{miner}(\mathbf{v} ; \mathbf{v}) \geq u_{miner}(\mu(\mathbf{v}) ; \mathbf{v}), 
    \end{equation}
    
    where a miner strategy function $\mu:R^*\rightarrow R^*$ is defined such that for any $\mathbf{v}\in R^*$ (i.e., a vector of bids of arbitrary length $n$), it outputs a vector $\mathbf{b}\in R^*$ of some length $n'\geq n$, so that for any $1\leq i \leq n$, $b_i \in \{0, v_i\}$.
    I.e., it has the choice whether to include or omit any of the original $n$ bids, as well as add $n' - n$ fake bids. Importantly, a miner cannot \textit{change} the real bidders' bids. 
\end{definition}

\subsubsection*{Collusion Notions}
A coalition of users and the miner may collude to increase their joint utility (as given in \cref{def:utilities}). 
Their collusion can entail any kind of deviation from the ``honest'' protocol, i.e., users may bid untruthfully, and the miner can deviate from the intended allocation rule.

The \gls{OCA}-proofness notion, introduced by \citeauthor{roughgarden2021transaction} \cite{roughgarden2021transaction}, compares the colluding coalition with the coalition comprised of the miner and the winning bidders of the canonical outcome, produced by following the mechanism's intended allocation.
The traditional definition given by \cite{roughgarden2021transaction} does not include a $c$ identifier for the amount of bidders in the coalition, but we add it since it allows considering how explicitly well-coordinated the collusion has to be in order to succeed.

\begin{definition}[\Glsxtrlong{OCA}-proof (\glsxtrshort{OCA}-proof)]
    \label{def:OCA}
    For a vector $\mathbf{v}\in R_+^n$ and $n'\geq n$, we say that $\Omega_{\mathbf{v}} \in R^{n'}$ is an off-chain agreement among a colluding coalition of the miner and a set of bidders $C$, if $\Omega_{\mathbf{v}}$ can change any of the colluding bidders' bids, as well as omit bids (whether by coalition members or not), and add $n' - n$ fake bids.
    For each agent $i$, we denote its bid under the agreement by $\Omega_{\mathbf{v}_i}$, and the bids of all other agents by $\Omega_{\mathbf{v}_{-i}}$.
    Let $c_{\Omega_{\mathbf{v}}}$ be the number of bids that $\Omega_{\mathbf{v}}$ \textit{changes} (rather than omits or keeps without change) such that they differ from the colluding bidders' true valuations, i.e., $c_{\Omega_{\mathbf{v}}} = \sum_{i=1}^n \mathds{1}_{\Omega_{\mathbf{v}_i} \not \in \{0,v_i\}}$.
    
    An auction is $c$-\gls{OCA}-proof if a colluding coalition of a miner and up to $c$ bidders cannot increase its aggregate utility to be higher than the joint utility of the intended allocation.
    Formally, for any valuation $\mathbf{v}$ vector of $n$ bidders, and any collusion $\Omega_{\mathbf{v}}$ of the miner and a set of bidders $C \subseteq [n]$ such that $|C| \leq c$ and $c_{\Omega_{\mathbf{v}}} \leq |C|$,
    %
    \begin{equation}
    \label{eq:oca_cond}
    u_{joint}(\mathbf{v} ;  \mathbf{v}) \geq u_{miner}(\Omega_{\mathbf{v}} ; \mathbf{v}) + \sum_{i \in C} u_i(\Omega_{\mathbf{v}_i}, \Omega_{\mathbf{v}_{-i}} ; v_i).\end{equation}
    If an auction is $c$-\gls{OCA}-proof for any $c$, we say it is \gls{OCA}-proof.
\end{definition}
\begin{remark}
    For ease of notation, our definition of $\Omega_{\mathbf{v}}$ incorporates all agent bids, including those who are not part of the colluding coalition.
    Thus, for a bid $b_i$ of an agent $i \notin C$, we have $\Omega_{\mathbf{v}_i} = b_i$.
    We emphasize that all bidders that \textit{change} their bid under $\Omega_{\mathbf{v}}$ (rather than keep it the same or omit it) \textit{must} be part of the colluding coalition.
    However, note that a collusion does not necessarily require that all colluding bidders modify their bids.
    This may matter for the design of transfers within the coalition, as \citeauthor{roughgarden2021transaction} \cite{roughgarden2021transaction} shows that any \gls{OCA} can be complemented with Pareto-improving transfers. 
\end{remark}

We show in \cref{clm:coalition_joint_bound} that the condition for $c$-\gls{OCA}-proofness takes a nice form when it holds for any value of $c$.
Due to space considerations, the proof is given in \cref{app:missing_proofs}.
\begin{restatable}[]{claim}{CoalitionJointBound}
\label{clm:coalition_joint_bound}
In an OCA-proof mechanism, for any coalition $C$ of bidders, $u_{joint}(\Omega_{\mathbf{v}} ; \mathbf{v}) \geq u_{miner}(\Omega_{\mathbf{v}} ; \mathbf{v}) + \sum_{i \in C} u_i(\Omega_{\mathbf{v}_i}, \Omega_{\mathbf{v}_{-i}} ; v_i).$
\end{restatable}

In \cref{lem:oca_joint}, we show that the characterization of \gls{OCA} provided by \citeauthor{roughgarden2021transaction} \cite{roughgarden2021transaction} holds under our refined definition, with the proof provided in \cref{app:missing_proofs}.
\begin{restatable}[]{lemma}{OCAJoint}
\label{lem:oca_joint}
A mechanism is \gls{OCA}-proof if and only if for any $\mathbf{v}$ and $\Omega_{\mathbf{v}}$,
\begin{equation}
\label{eq:general_oca}
u_{joint}(\mathbf{v} ; \mathbf{v}) \geq u_{joint}(\Omega_{\mathbf{v}} ; \mathbf{v}). 
\end{equation}
\end{restatable}

The line of works started by \citeauthor{chung2023foundations} \cite{chung2023foundations} use the \gls{SCP} collusion notion, which compares the aggregate utility of the colluding coalition with the same coalition's ``honest'' aggregate utility, i.e., the coalition's utility if they were to act honestly.
\begin{definition}[\Glsxtrfull{SCP}]
    \label{def:Scp}
    For a vector $\mathbf{v}\in R_+^n$ and $n'\geq n$, we say that $\Gamma_{\mathbf{v}} \in R^{n'}$ is a side-contract among a colluding coalition of the miner and a set of bidders, if $\Gamma_{\mathbf{v}}$ can change any of the colluding bidders' bids, as well as omit bids (whether by coalition members or not), and add $n' - n$ fake bids.
    For each agent $i$, we denote its bid under the contract by $\Gamma_{\mathbf{v}_i}$, and the bids of all other agents by $\Gamma_{\mathbf{v}_{-i}}$.
    Let $c_{\Gamma_{\mathbf{v}}}$ be the number of bids that $\Gamma_{\mathbf{v}}$ \textit{changes} (rather than omits or keeps without change) such that they differ from the colluding bidders' true valuations, i.e., $c_{\Gamma_{\mathbf{v}}} = \sum_{i=1}^n \mathds{1}_{\Gamma_{\mathbf{v}_i} \not \in \{0,v_i\}}$.

    An auction is $c$-\gls{SCP}-proof if a coalition of a miner and up to $c$ bidders cannot increase their aggregated utility by deviating from the honest protocol.
    Formally, for any valuation vector $\mathbf{v}$ of $n$ bidders, and any collusion $\Gamma_{\mathbf{v}}$ of the miner and a set of bidders $C \subseteq [n]$ such that $|C| \leq c$ and $c_{\Gamma_{\mathbf{v}}} \leq |C|$,
    %
    \begin{equation}
        \label{eq:scp_cond} 
        u_{miner}(\mathbf{v} ; \mathbf{v}) + \sum_{i \in C} u_i(v_i, \mathbf{v}_{-i} ; v_i) \geq u_{miner}(\Gamma_{\mathbf{v}} ; \mathbf{v}) + \sum_{i \in C} u_i(\Gamma_{\mathbf{v}_i}, \Gamma_{\mathbf{v}_{-i}} ; v_i).
    \end{equation}
    If an auction is $c$-\gls{SCP} for any $c$, we say it is \gls{SCP}. 
\end{definition}

\section{Understanding OCA}

\subsection{The Difference Between SCP and OCA}
We observe that while \gls{OCA}-proofness compares the utility of every possible manipulating coalition with the joint utility of the winning coalition under the protocol's intended allocation, \gls{SCP} compares it with the counter-factual of the same coalition's honest utility (i.e., the utility obtained without any manipulation).
We formally prove this distinction in \cref{clm:OcaVsScp}. 

\begin{claim}
    \label{clm:OcaVsScp}
    \gls{OCA}-proof $\not\Rightarrow$ $1$-\gls{SCP}.
\end{claim}
\begin{proof}
    Consider $3$ agents with values $\mathbf{v} = (1, \frac{1}{2}, \frac{1}{4})$ that participate in a \emph{third-price} auction:
    \begin{definition}[Third-price Auction]
        The auction is defined as follows:
        
        \noindent {\bf Allocation rule:} allocate the item to the highest bidder: $\tilde{\alloc}(\mathbf{b}) = \arg \max \mathbf{b}.$
        
        \noindent {\bf Payment rule:} the winning bidder pays the $3^{rd}$ highest bid. I.e., if we let $b_{3}$ be the third highest element in $\mathbf{b}$, then $\tilde{\pay}(\mathbf{b}) = b_{3}$.
        
        \noindent {\bf Burning rule:} no fees are burnt. $\tilde{\burn}(\mathbf{b}) = 0$.
    \end{definition}
    If all agents bid truthfully, this auction allocates the item (e.g., the transaction) to the first bidder, and requires that the winner pays the third-highest bid, i.e., $\frac{1}{4}$.
    No payment is burned.

    The auction is not $1$-\gls{SCP}.
    To see why, consider the coalition of the miner together with bidder $2$. Then, without any side contract agreement, their joint utility is $\frac{1}{4}$: the miner's payment is $\frac{1}{4}$ and bidder $2$'s utility is $0$, since its transaction is not allocated:

    \[
    \begin{split}
    u_{miner}(\mathbf{v} ; \mathbf{v}) + u_2(v_2, \mathbf{v}_{-2} ; v_2)
    &
    =
    \sum_{i=1}^3 \left(\pay(v_i, \mathbf{v}_{-i}) - \burn(v_i, \mathbf{v}_{-i}) \right) +  v_2 \cdot \alloc(v_2, \mathbf{v}_{-2}) - \pay\left(v_2,\mathbf{v_{-2}}\right)
    \\&
    =
    \pay(v_1, \mathbf{v}_{-1})
    =
    \frac{1}{4}.
\end{split}
\]

    However, if we consider the side contract agreement $\Gamma_{\mathbf{v}} = (1,2,\frac{1}{4})$ with $c_{\Gamma_{\mathbf{v}}} = 1$, where bidder $2$ submits the bid $2$ instead of $\frac{1}{2}$, then the coalition's utility changes to $\frac{1}{2}$: The miner's payment is still $\frac{1}{4}$, and bidder $2$ has a utility of $\frac{1}{2} - \frac{1}{4} = \frac{1}{4}$:

    \[
    \begin{split}
    u_{miner}(\Gamma_{\mathbf{v}} ; \mathbf{v}) + u_2((\Gamma_{\mathbf{v}_2}, \Gamma_{\mathbf{v}_{-2}} ; v_2)
    \hspace*{.3em}
    &
    \overset{\mathclap{\text{\tiny Eqs. \ref{eq:miner_util},\ref{eq:bidder_util}}}}
    =
    \hspace*{.3em}
    \sum_{i=1}^3 \left(\pay(\Gamma_{\mathbf{v}_i}, \Gamma_{\mathbf{v}_{-i}}) - \burn(\Gamma_{\mathbf{v}_i}, \Gamma_{\mathbf{v}_{-i}}) \right) + v_2 \cdot \alloc(\Gamma_{\mathbf{v}_2}, \Gamma_{\mathbf{v}_{-2}}) - \pay\left(\Gamma_{\mathbf{v}_2}, \Gamma_{\mathbf{v}_{-2}})\right)
    \\&
    =
    \hspace*{.3em}
    \pay(\Gamma_{\mathbf{v}_2}, \Gamma_{\mathbf{v}_{-2}}) + v_2 \cdot \alloc(\Gamma_{\mathbf{v}_2}, \Gamma_{\mathbf{v}_{-2}}) - \pay\left(\Gamma_{\mathbf{v}_2}, \Gamma_{\mathbf{v}_{-2}})\right)
    \\&
    =
    \hspace*{.3em}
    v_2 \cdot \alloc(\Gamma_{\mathbf{v}_2}, \Gamma_{\mathbf{v}_{-2}})
    =
    \frac{1}{2}.
\end{split}
\]

    On the other hand, the auction is \gls{OCA}-proof, since for any $\mathbf{v}$ and \gls{OCA} $\Omega_{\mathbf{v}}$ it holds that

    $$
        u_{joint}(\mathbf{v} ; \mathbf{v})
        =
        \max \mathbf{v}
        \overset{\mathclap{\text{\tiny\cref{clm:simple_joint_bound}}}}
        \geq
        u_{joint}(\Omega_{\mathbf{v}} ; \mathbf{v}),
    $$

    and this is exactly equivalent to being OCA-proof by \cref{lem:oca_joint}. 
    

\end{proof}

We now prove \cref{clm:ScpImpliesOca}, which together with \cref{clm:OcaVsScp} implies that \gls{SCP} is stricter than \gls{OCA}.
\begin{claim}
    \label{clm:ScpImpliesOca}  
    $c$-\gls{SCP} $\Rightarrow$ $c$-\gls{OCA}-proof.
\end{claim}
\begin{proof}
    Let $\auction$ be an auction which is $c$-\gls{SCP}. Then for any $\mathbf{v}$ and corresponding off-chain agreement $\Omega_{\mathbf{v}}$ and coalition $C$ so that $|C| \leq c$, choose the side contract agreement $\Gamma_{\mathbf{v}} = \Omega_{\mathbf{v}}$ to be the same as the off-chain agreement, we then have:

    $$   u_{joint}(\mathbf{v} ; \mathbf{v}) \stackrel{\text{Claim~\ref{clm:coalition_joint_bound}}}{\geq} u_{miner}(\mathbf{v} ; \mathbf{v}) + \sum_{i \in C} u_i(v_i, \mathbf{v}_{-i} ; v_i) \stackrel{\text{\cref{eq:scp_cond}}}{\geq} u_{miner}(\Omega_{\mathbf{v}} ; \mathbf{v}) + \sum_{i \in C} u_i(\Omega_{\mathbf{v}_i}, \Omega_{\mathbf{v}_{-i}} ; v_i), $$

    but this is exactly the $c$-OCA-proofness condition of \cref{eq:oca_cond}.      
\end{proof}

Now that we distinguished SCP and OCA-proofness, we turn our focus to characterizing and obtaining results for \gls{TFM}s that satisfy OCA-proofness.

\subsection{Useful Preliminary Results for OCA-proof TFMs}

In this section, we derive results that hold for any (possibly randomized) mechanism, and without further qualifiers which we define later such as scale-invariant mechanisms. 
First, we state Myerson's lemma for DSIC mechanisms. 

\begin{fact}[Myerson's lemma, as given in \cite{hartline2006lectures})]
    \label{fct:MyersonLemma}
    A single-item auction is \gls{DSIC} if and only if for all $i$ and $\mathbf{b}_{-i}$, the allocation $\alloc(b_i, \mathbf{b}_{-i})$ is monotone and non-decreasing in $b_i$, and furthermore the payment is exactly $\pay(b_i, \mathbf{b}_{-i}) = \alloc(b_i, \mathbf{b}_{-i}) \cdot b_i - \int_0^{b_i}\alloc(t, \mathbf{b}_{-i}) dt$.
\end{fact}

We continue by reasoning about the single-bidder case.
\begin{restatable}[]{lemma}{SingleBidderRevenue}
\label{lem:SingleBidderRevenue}
    With a single bidder, all \gls{DSIC}+$1$-OCA-proof single-item auctions have $0$ seller revenue. 
\end{restatable}

The above statement works for the single bidder case, but does not extend to more bidders.
We demonstrate it through the example of the second-price auction:


\begin{example}
\label{ex:second_price}
    \textit{The second-price auction is DSIC and $1$-OCA-proof (as observed in \cite{roughgarden2021transaction})}. First, notice that in the single bidder case, the second-price auction indeed satisfies that the payment rule equals the burn rule (both are always $0$). However, with more bidders, the form of the joint utility of the winner coalition no longer resembles the form of the utility of a specified bidder, since it takes the value of the \textit{maximal} bidder, depending on $\mathbf{b}$ (rather than any fixed bidder: This also separates the analysis of OCA-proofness from that of SCP, for which we know that DSIC+$1$-SCP indeed yields $0$ miner revenue). 
\end{example}

\section{Deterministic OCA-proof Mechanisms}

\cref{ex:second_price} shows that generally, the DSIC and $1$-OCA-proofness properties are not enough to guarantee zero revenue.
We now show that for deterministic mechanisms, adding the MMIC property suffices to get a general $0$ revenue result. 
%
\begin{theorem}
\label{thm:zero_revenue}
Every deterministic \gls{DSIC}+\gls{MMIC}+$1$-\gls{OCA}-proof mechanism has $0$ miner revenue.
\end{theorem}
\begin{proof}

Assume towards contradiction a bid vector $\mathbf{b} = b_1, \ldots, b_n$ such that $u_{miner}(\mathbf{b} ; \mathbf{b}) > 0$. 

This means, by \cref{eq:miner_util},
 $u_{miner} (\mathbf{b} ; \mathbf{b}) =  \tilde{\pay}(\mathbf{b}) - \tilde{\burn}(\mathbf{b}) > 0.$

Since the payment is strictly positive, by individual rationality (\cref{eq:IR}), the item is allocated to some bidder $i$. All other bidders have no payments or burns. Thus,  again by \cref{eq:miner_util},

$$
u_{miner} (\mathbf{b} ; b_i) = \pay(b_i, \mathbf{b}_{-i}) - \burn(b_i, \mathbf{b}_{-i}) - \sum_{j\neq i} \burn(b_j, \mathbf{b}_{-j}) = \tilde{\pay}(\mathbf{b}) - \tilde{\burn}(\mathbf{b}) > 0 = u_{miner}(b_i ; b_i).
$$

We thus have a violation of the MMIC condition of \cref{eq:mmic_cond}, where $\mathbf{v} = (b_i), \mu(\mathbf{v}) = \mathbf{b}$, i.e., when the miner is faced with the single bid $b_i$, it might insert the fake transactions $\mathbf{b}_{-i}$.
\end{proof}

The characterization of \cref{lem:SingleBidderRevenue}, which holds in the general randomized case, gives the following result in the \textit{deterministic} case:

\begin{corollary}
In the single bidder case, any deterministic, DSIC+$1$-OCA-proof mechanism $\alloc, \pay, \burn$ has an allocation function of the form $\tilde{a}(b_1) = \begin{cases} 1 & b_1 \geq r, \\
\emptyset & Otherwise\end{cases}$,  and payment and burn functions of the form $\tilde{\pay}(b_1) = \tilde{\burn}(b_1) = \mathds{1}_{b_1 \geq r} \cdot r$, for some $r \in R_+$.
\end{corollary}

However, we can provide a meaningful characterization even when removing the \gls{DSIC} condition.
The characterization, given in \cref{lem:oca_single_bidder}, remains very similar, albeit with more freedom to decide the payment rule. 

\begin{lemma}
\label{lem:oca_single_bidder}
In the single bidder case, in any $1$-\gls{OCA}-proof deterministic mechanism $\alloc, \pay, \burn$, the allocation function takes the form $\tilde{a}(b_1) = \begin{cases} 1 & b_1 \geq r, \\
\emptyset & Otherwise\end{cases}$, and the burn function takes the form $\tilde{\burn}(b_1) = \mathds{1}_{b_1 \geq r} \cdot r$ for some $r \in R_+$.
\end{lemma}

\begin{proof}
Let $r = \inf \{b_1 \text{ s.t. } \tilde{\alloc}(b_1) \neq \emptyset \}$. For any $b_1 > r$, assume towards contradiction that $\tilde{\alloc}(b_1) = \emptyset$. Then, there must be some $b'_1$ with $r < b'_1 < b_1$ so that $\tilde{\alloc}(b'_1) \neq \emptyset$, as $r$ is the infimum for this property. Then, $\tilde{\burn}(b'_1) \stackrel{\text{Eqs. (}\ref{eq:IR},\ref{eq:BB}\text{)}}{\leq} b'_1 < b_1$, and so 

$$u_{joint}(b'_1 ; b_1) = b_1 \tilde{\alloc}(b'_1) - \tilde{\burn}(b'_1) = b_1 - \tilde{\burn}(b'_1) > 0 = b_1 \tilde{\alloc}(b_1) - \tilde{\burn}(b_1) = u_{joint}(b_1 ; b_1),$$

which violates the condition of \cref{lem:oca_joint} (which is the same as $1$-OCA-proofness since we discuss a single bidder), where $\mathbf{v} = (b_1), \Omega_{\mathbf{v}} = (b'_1)$. 

For any $b_1 < r$ we know by the choice of $r$ as the infimum of the property, that $\tilde{\alloc}(b_1) = \emptyset$.

Regarding the burn, consider that there are some two bids $b_1 \neq b'_1$ so that $\tilde{\alloc}(b_1) = \tilde{\alloc}(b'_1) = 1$, but the burn is different between the two, w.l.o.g. assume that $\tilde{\burn}(b_1) > \tilde{\burn}(b'_1)$. This violates the OCA-proofness condition of \cref{lem:oca_joint}, where $\mathbf{v} = (b_1), \Omega_{\mathbf{v}} = (b'_1)$, as:

$$u_{joint}(b'_1 ; b_1) = b_1 \tilde{\alloc}(b'_1) - \tilde{\burn}(b'_1) = b_1 - \tilde{\burn}(b'_1) > b_1 - \tilde{\burn}(b_1) = b_1 \tilde{\alloc}(b_1) - \tilde{\burn}(b_1) = u_{joint}(b_1 ; b_1).$$

We conclude that the burn for all allocated values is some constant $R$. We now compare $R$ with the $r$ we have for the allocation rule. 

If $R < r$, compare $b_1 = \frac{R+r}{2}$ (where it holds that $R < b_1 < r$), with some $b'_1 > r$. I.e., consider the $1$-OCA where $\mathbf{v} = (b_1)$, and $\Omega_{\mathbf{v}} = (b'_1)$. Since $b_1 < r$, $\tilde{\alloc}(b_1) = 0$. We have:
$$u_{joint}(b'_1 ; b_1) = b_1 \tilde{\alloc}(b'_1) - \tilde{\burn}(b'_1) = b_1 - R > 0 = b_1 \tilde{\alloc}(b_1) - \tilde{\burn}(b_1) = u_{joint}(b_1 ; b_1).$$

If $R > r$, then for some $r < b_1 < R$, 
$$u_1(b_1 ; b_1) = b_1 \tilde{\alloc}(b_1) - \tilde{\pay}(b_1) \stackrel{\text{\cref{eq:BB}}}{\leq} b_1 \tilde{\alloc}(b_1) - \tilde{\burn}(b_1) = b_1 - R < 0,$$
which violates \cref{eq:IR_util}.

We conclude that $R = r$, which yields the specified characterization. 
\end{proof}

This allows us to further characterize the allocation and burn rules more generally, for deterministic $1$-OCA-proof mechanisms. 

\begin{lemma}
\label{lem:general_oca_char}
Any $1$-OCA-proof deterministic mechanism $\alloc, \pay, \burn$ is exactly of the following form: 
For some $r \geq 0$, the mechanism allocates the item to the highest bidder subject to it having higher value than $r$, or does not allocate the item at all. Whenever allocated, the burn is exactly $r$. 
I.e., $$\tilde{\alloc}(\mathbf{b}) = \begin{cases} \arg  \max \mathbf{b} & \max \mathbf{b} \geq r \\
\emptyset & Otherwise,\end{cases}, \tilde{\burn}(\mathbf{b}) = r.$$
\end{lemma}

\begin{proof}
($\implies$ All $1$-OCA-proof mechanisms must be of the form specified)

The single bidder case follows the characterization in \cref{lem:oca_single_bidder}. Now consider a general valuation vector $\mathbf{v}$ that is allocated to agent $i$ with burn $R$. It must hold that $R=r$, or we would have a $1$-OCA between agent $i$ and the miner, where the miner either removes all the other bidders (if $R > r$), or adds fake bidders in the single bidder case (if $r > R$). We conclude that whenever the item is allocated, the burn must be exactly $r$. Now consider if the item is allocated to agent $j$ so that $j \neq \arg \max \mathbf{v}$ (and let $i = \arg \max \mathbf{v}$). Then, the joint utility is exactly $v_j - r$. But then there is a $1$-OCA between agent $i$ and the miner where the miner removes all bidders besides $i$ (including $j$), and we have a joint utility of $v_i - r$. Thus, we conclude that \textit{if} the item is allocated, it must be allocated to the highest bidder. 

The item cannot be allocated if the highest bidder $i$ has $b_i < r$, because then the burn must be lower than $r$, which allows for a $1$-OCA (for example for a single bidder with $b'_i > r$ and where the miner adds fake bids of $\mathbf{b}_{-i}$). If $b_i > r$, then the item must be allocated, or there is a $1$-OCA where the miner omits all other bids and the utility is $b_i - r$ by the single bidder case. 

($\implies$ All auctions of the specified form are $1$-OCA-proof)

Consider $\mathbf{v}$ so that the item is allocated and let $i = \arg \max \mathbf{v}$. Then, the joint utility is $v_i -r \geq 0$, and so moving to another set of bids $\mathbf{b}$ where the item is not allocated reduces the joint utility to $0$. Moving to $\mathbf{b}$ so that the item is allocated to another agent $j$ results in joint utility $v_j - r$ (since the burn is constantly $r$), which is also reduced. Consider $\mathbf{v}$ so that the item is not allocated, then moving to another set of bids $\mathbf{b}$ so that the item is not allocated does not change the joint utility, and if the item \textit{is} allocated to some agent $j$, the joint utility is $v_j - r \leq v_i - r \leq 0$ (as we know $v_i \leq r$ since the item is not allocated in the original $\mathbf{v}$). 

\end{proof}

We now can precisely characterize two classes of mechanisms: The class of DSIC+$1$-OCA-proof deterministic mechanisms, and the class of MMIC+$1$-OCA-proof deterministic mechanisms.

\begin{restatable}[]{theorem}{DSICOCA}
\label{thm:dsic_1oca}
The class of DSIC+$1$-OCA-proof deterministic mechanisms is exactly the class of all auctions with reserve $r$ that allocate to the highest bidder and which burn $r$ from the allocated bid.
I.e., for some $r\geq 0$, the allocation rule is $\tilde{\alloc}(\mathbf{b}) = \begin{cases} \arg  \max \mathbf{b} & \max \mathbf{b} \geq r, \\
\emptyset & Otherwise\end{cases}$, the payment rule is $\tilde{\pay}(\mathbf{b}) = \max (
\{ r\} \cup (\mathbf{b} \setminus \{\max \mathbf{b}\}))$, and the burn rule is $\tilde{\burn}(\mathbf{b}) = r$.
\end{restatable}

\begin{restatable}[]{theorem}{MMICOCA}
    
\label{thm:mmic_1oca}
The class of MMIC+$1$-OCA-proof deterministic mechanisms is exactly the class of all second-price auctions with reserve $r$, which always burn $r$ when allocated. I.e., for some $r\geq 0$, $$\tilde{\alloc}(\mathbf{b}) = \begin{cases} \arg  \max \mathbf{b} & \max \mathbf{b} \geq r \\
\emptyset & Otherwise,\end{cases}, \tilde{\pay}(\mathbf{b}) = f(\max \mathbf{b}) , \tilde{\burn}(\mathbf{b}) = r,$$

for some monotone function $f(v) \geq r$. 
\end{restatable}

These precise characterizations now allow us to conclude with the following:
\begin{theorem}
\label{thm:det_impossibility}
Never allocating the item is the only DSIC+MMIC+$1$-OCA-proof deterministic mechanism. 
\end{theorem}
\begin{proof}
This follows from \cref{thm:dsic_1oca} and \cref{thm:mmic_1oca}, as the two classes characterized in these results only have the trivial mechanism in common (taking $r = \infty$).
To intuitively see this, consider the class of second-price auctions with reserve $r$ and constant burn $r$ of \cref{thm:dsic_1oca}.
Second-price auctions are not MMIC since the miner can add a fake bidder arbitrarily close to the winning bid to increase the payment. 
\end{proof}

\section{Randomized OCA-proof Mechanisms}

We now extend the discussion to \textit{randomized} OCA-proof mechanisms.
For randomized mechanisms, we consider the stronger notion of OCA-proofness (rather than $1$-OCA-proofness).
We do so to avoid clutter in the definitions, as in randomized mechanisms the winning coalition may very well \textit{necessarily} include \textit{all} bidders (as each has some fractional probability of winning). 

We now consider a natural property for mechanisms:
\begin{definition}[Scale Invariance]
We say that a mechanism $a, p, \beta$ is scale-invariant if for any set of bids $\mathbf{b}$, and a constant $\alpha > 0$, it holds that $a(b_i, \mathbf{b}_{-i}) = a(\alpha b_i, \alpha \mathbf{b}_{-i})$.
I.e., re-scaling all the bids by the same factor does not alter the allocation.
\end{definition}


\begin{restatable}[]{claim}{homogeneousPayments}
For a DSIC, scale-invariant mechanism $\alloc, \pay, \burn$, it holds that: 
$$p(\alpha b_i, \alpha \mathbf{b}_{-i}) = \alpha p(b_i, \mathbf{b}_{-i}).$$
\end{restatable}

\begin{lemma}
\label{lem:zero_burns}
For any set of bids $\mathbf{b}$, a scale-invariant mechanism that is both \gls{DSIC} and \gls{OCA}-proof has $\beta(b_i, \mathbf{b}_{-i}) = 0$.
\end{lemma}
\begin{proof}
Assume toward contradiction that there is a set of bids $\mathbf{b}$ such that $\bar{\beta} \stackrel{def}{=} \sum_i \beta(b_i, \mathbf{b}_{-i}) > 0$.
The burn is upper-bounded by the bid itself, i.e. $\beta(b_i, \mathbf{b}_{-i}) < b_i$, due to both individual rationality and the \gls{DSIC} property.
Thus, if we choose $\alpha = \frac{\bar{\beta}}{2\sum_i b_i}$, we are guaranteed that $\sum_i \beta(\alpha b_i, \alpha \mathbf{b}_{-i}) \leq \sum_i \alpha b_i = \frac{1}{2} \bar{\beta}$.
Since the allocation is the same by the scale-invariance property, the joint utility must increase due to lower burns.
This violates \gls{OCA}-proofness, where we consider the off-chain agreement that all bidders scale down their bids by factor $\alpha$. 
\end{proof}

\begin{corollary}
    \label{cor:scale_invariant_payment}
    By \cref{lem:zero_burns}, a DSIC+OCA-proof scale-invariant mechanism does not burn fees (i.e., its burn rule is the constant zero function), while from \cref{lem:SingleBidderRevenue} we get that a DSIC+MMIC+OCA-proof mechanism has payments equal to the burn in the single bidder case. Therefore, we must have $0$ payments in the single bidder case, and so, in the single bidder case, the item is either always or never allocated.
\end{corollary}

\begin{lemma}
\label{lem:always_or_never_allocated}
For a DSIC+MMIC+OCA-proof mechanism, if the item is always or never allocated in the single bidder case, the mechanism must be trivial.
\end{lemma}

\begin{proof}
First, consider the case where the item is always allocated in the single bidder case. Then, by an earlier claim, for any vector of valuations $\mathbf{v}$, 

$$u_{joint}(\mathbf{v} ; \mathbf{v}) \stackrel{\text{Claim~\ref{clm:simple_joint_bound}}}{\leq} \max \mathbf{v} = u_{joint}(\max \mathbf{v} ; \mathbf{v}). $$

If the inequality is strict, this would violate the OCA condition of \cref{lem:oca_joint}. Thus, we must always have $$u_{joint}(\mathbf{v} ; \mathbf{v}) = \max \mathbf{v},$$
i.e., the item must be allocated in full to the highest bidder. However, this brings us back to the deterministic discussion, and the impossibility of \cref{thm:det_impossibility}.

Now consider the case where the item is never allocated in the single bidder case. If the item has non-zero allocation for some bidder $i$ with some valuation vector $\mathbf{v}$, then

$$u_{joint}(\mathbf{v} ; v_i) \stackrel{\text{\cref{eq:joint_util}}}{=} v_i \alloc(v_i, \mathbf{v}_{-i}) - \beta(v_i, \mathbf{v}_{-i}) - \sum_{j\neq i} \beta(v_j, \mathbf{v}_{-j}) \stackrel{\text{\cref{lem:zero_burns}}}{=} v_i \alloc(v_i, \mathbf{v}_{-i}) > 0 = u_{joint}(v_i ; v_i),$$

violating the OCA-proofness condition of \cref{lem:oca_joint}.
\end{proof}

Thus, as a direct result of Corollary~\ref{cor:scale_invariant_payment} and \cref{lem:always_or_never_allocated}, we have:

\begin{corollary}
\label{thm:scale_invariant}
There is no non-trivial scale-invariant DSIC+MMIC+OCA-proof mechanism.   
\end{corollary}

The argument we use in \cref{lem:always_or_never_allocated} can be extended to allow us to also rule out the class of auctions that satisfy a property that we call \emph{\gls{CTPA}}, which is defined in \cref{def:CTPA}. 
This is an interesting class of auctions, as it includes all \textit{efficient} auctions (that are part of the class of constant total probability $1$ of allocation), including the first-price and second-price auctions.
%
\begin{definition}[Constant Total Probability of Allocation]
\label{def:CTPA}
We say a mechanism $a, p, \beta$ is \gls{CTPA} if for any set of bids $\mathbf{b}$, we have 
$\sum_i a(b_i, \mathbf{b}_{-i}) = \alpha$
for some constant $0 \leq \alpha \leq 1$. 
\end{definition}
As we show in \cref{thm:ctpa}, the \gls{CTPA} property also suffices to rule out any mechanism other than the trivial one. The proof uses similar ideas as these of Lemma~\ref{lem:always_or_never_allocated}.
\begin{restatable}[]{theorem}{CTPA}
\label{thm:ctpa}
There is no non-trivial CTPA, DSIC+MMIC+OCA-proof mechanism. 
\end{restatable}

\begin{corollary}
\label{corr:no_efficient}
There is no \textit{efficient} DSIC+MMIC+OCA-proof mechanism. 
\end{corollary}

Now that we know that an efficient mechanism cannot be achieved, we would like to further bound the efficiency level that we can achieve away from perfect welfare.

We start by proving a useful claim:
\begin{claim}
\label{clm:low_value_feasibility}
Let $\alloc, \pay, \burn$ be an OCA-proof mechanism. 
Let $\mathbf{b}$ be a set of $n$ bids, and let $b_1 = \min \mathbf{b}$. Then, $\alloc(b_1, \mathbf{b}_{-1}) \leq \frac{1}{n} \sum_i \alloc(b_i, \mathbf{b}_{-i})$.
\end{claim}

\begin{proof}
Assume towards contradiction that $\alloc(b_1, b_{-1}) > \frac{1}{n} \sum_i a(b_i, b_{-i})$. Then, there must be some $j\neq 1$, so that $b_j \geq b_1$ and $\alloc(b_j, b_{-j}) < \frac{1}{n} \sum_i a(b_i, b_{-i}) < \alloc(b_1, b_{-1})$. By anonymity, if we let bidders $1$ and $j$ bid each other's bid, the expected burn does not change, but we now allocate more to a higher value bidder and less to a lower value bidder, while not affecting all other bidders. This increases the joint utility and thus violates $1$-OCA-proofness.
\end{proof}


\begin{claim}
\label{clm:payment_burn_bound}
With a DSIC, MMIC and OCA-proof mechanism, it holds that for two bids $b_1, b_2$:

$$\pay(b_1, b_2) \leq \burn(b_1, b_2) + \burn(b_2, b_1).$$
\end{claim}

\begin{proof}
Assume towards contradiction there are some $b_1, b_2$ so that $\pay(b_1, b_2) > \burn(b_1, b_2) + \burn(b_2, b_1)$. We know by \cref{lem:SingleBidderRevenue} that $u_{miner}(b_1 ; b_1) = 0$. Consider the miner strategy $\mu_{\mathbf{v}} = (b_1, b_2)$ for $\mathbf{v} = (b_1)$ (i.e., the miner adds a fake bidder $b_2$). Then, the miner utility is
$$u_{miner}((b_1, b_2) ; b_1) = \pay(b_1, b_2) - \burn(b_1, b_2) - \burn(b_2, b_1) > 0 = u_{miner}(b_1 ; b_1),$$

which violates the MMIC condition of \cref{eq:mmic_cond}.

\end{proof}

\begin{lemma}
\label{lem:two_bidders_cond}
With a DSIC, MMIC and OCA-proof mechanism, it holds that for two bids $b_1 \geq b_2$:

$$\alloc(b_2, b_1) \cdot b_2 + \int_0^{b_1} a(t, b_2) dt \geq \int_0^{b_1} a(t) dt. $$

\end{lemma}

\begin{proof}
By the DSIC property, through Fact~\ref{fct:MyersonLemma}, we know that the allocation function is monotonically increasing in a bidder's bid, and we have an exact formula for the payment given an allocation function.
By \cref{lem:SingleBidderRevenue}, we know that in the single bidder case, the miner revenue is always $0$, implying that the payment and burn functions are equal. We can thus write:
\begin{equation}
    \begin{split}
        \label{eq:rhs_joint}
        u_{joint}(b_1 ; b_1, b_2)
        =
        a(b_1) \cdot b_1 - \beta(b_1) 
        =
        a(b_1) \cdot b_1 - \pay(b_1) 
        =
        \int_0^{b_1} a(t) dt.
    \end{split}
\end{equation}

We can also write 
\begin{equation}
    \begin{split}
        \label{eq:lhs_joint}
        u_{joint}(b_1, b_2 ; b_1, b_2)
        &
        =
        a(b_1, b_2) \cdot b_1 - \beta(b_1, b_2) + a(b_2, b_1) \cdot b_2 - \beta(b_2, b_1)
        \\&
        \stackrel{\text{Claim~\ref{clm:payment_burn_bound}}}{\leq} 
        a(b_1, b_2) \cdot b_1 - \pay(b_1, b_2) + a(b_2, b_1) \cdot b_2 
        \\&
        \stackrel{\text{Fact~\ref{fct:MyersonLemma}}}{=}
        a(b_2, b_1) \cdot b_2  + \int_0^{b_1} a(t, b_2) dt.
    \end{split}
\end{equation}

Thus, by the OCA condition of \cref{lem:oca_joint},

$$a(b_2, b_1) \cdot b_2  + \int_0^{b_1} a(t, b_2) dt \stackrel{\text{\cref{eq:lhs_joint}}}{\geq}  u_{joint}(b_1, b_2 ; b_1, b_2) \stackrel{\text{\cref{lem:oca_joint}}}{\geq} u_{joint}(b_1 ; b_1, b_2) \stackrel{\text{\cref{eq:rhs_joint}}}{=} \int_0^{b_1} a(t) dt. $$

\end{proof}

\begin{theorem}
\label{thm:allocation_bound}
Any DSIC+MMIC+OCA-proof mechanism for a single bidder, allocates the item with probability at most $\sqrt{2} - \frac{1}{2} \approx 0.91421356$. 
\end{theorem}

\begin{proof}

Let $b$ be some single bid, and let $\alpha = a(b)$. Consider the valuation vector $v^* = (A \cdot b, A \cdot b)$ for some parameter $A > 1$, and let $B$ so that $1 < B < A$. We wish to compare the joint utility of reporting $v^*$ truthfully, versus the off-chain agreement where the miner omits one of the bidders' bid. We get by \cref{lem:two_bidders_cond}:

\begin{equation}
\label{eq:lower_bound_A1}
\alloc(A \cdot b, A \cdot b) \cdot A \cdot b + \int_0^{A\cdot b} a(t, A \cdot b) dt \geq \int_0^{A \cdot b} a(t) dt \geq \int_b^{A \cdot b} a(t) dt \geq \alpha \cdot (A-1) \cdot b.
\end{equation}

By applying Claim~\ref{clm:low_value_feasibility} to the two bidder case, we get:
\begin{equation}
\label{eq:upper_bound_AB}
\alloc(A \cdot b, A \cdot b) \cdot A \cdot b  + \int_{B \cdot b}^{A \cdot b}  a(t, A \cdot b)dt \leq \frac{1}{2} A \cdot b + \int_{B \cdot b}^{A \cdot b}  \frac{1}{2} dt = (A-\frac{1}{2}B) \cdot b.\end{equation}

Together, \cref{eq:lower_bound_A1} and \cref{eq:upper_bound_AB} yield:
\begin{equation}
\begin{split}
& a(B \cdot b, A \cdot b) \cdot B \cdot b = \int_{0}^{B \cdot b}  a(B \cdot b, A \cdot b)dt \geq
\int_{0}^{B \cdot b}  a(t, A \cdot b)dt = \\
& \left(\alloc(A \cdot b, A \cdot b) \cdot A \cdot b + \int_{0}^{A \cdot b}  a(t, A \cdot b)dt \right) - \left(\alloc(A \cdot b, A \cdot b) \cdot A \cdot b + \int_{B \cdot b}^{A \cdot b}  a(t, A \cdot b)dt \right) \geq \\
& \left( \alpha (A-1) - (A-\frac{1}{2}B) \right) \cdot b.
\end{split}\end{equation}


By rearranging the terms we get:
\begin{equation}a(B\cdot b, A \cdot b) \geq \frac{\alpha (A-1) - (A-\frac{1}{2}B) }{B},\end{equation}
and thus by the feasibility \cref{eq:feasibility}:
\begin{equation} a(A\cdot b, B \cdot b) \leq 1 - \frac{\alpha (A-1) - (A-\frac{1}{2}B) }{B} = \frac{1}{2} + \frac{(1-\alpha)A + \alpha}{B}. \end{equation}

This yields:
\begin{equation} \int_0^{A \cdot b} a(t, B \cdot b)dt \leq \int_0^{A \cdot b} a(A \cdot b, B \cdot b)dt = \left(\frac{1}{2} + \frac{(1-\alpha)A + \alpha}{B} \right) \cdot A \cdot b. \end{equation}

Also, by Claim~\ref{clm:low_value_feasibility}, we have:
\begin{equation} 
\alloc(B\cdot b, A\cdot b) \cdot B \cdot b \leq \frac{B \cdot b}{2}. \end{equation}

Hence,
\begin{equation}
\alloc(B\cdot b, A\cdot b) \cdot B \cdot b + \int_0^{A \cdot b} a(t, B \cdot b)dt \leq 
\left( \left(\frac{1}{2} + \frac{(1-\alpha)A + \alpha}{B} \right) \cdot A + \frac{B}{2} \right) \cdot b.
\end{equation}

Notice that this is the left-hand side of the \cref{lem:two_bidders_cond} where we consider the bids $B \cdot b, A \cdot b$. We can thus repeat the way we developed \cref{eq:lower_bound_A1} (for the case of the bids $A \cdot b, A \cdot b$) and, by considering that the miner omits the bid $B \cdot b$, get:
\begin{equation}
\alloc(B\cdot b, A\cdot b) \cdot B \cdot b + \int_0^{A \cdot b} a(t, B \cdot b)dt \geq \int_0^{A\cdot b} a(t) dt \geq \alpha \cdot (A-1) \cdot b.
\end{equation}
Overall, the yields:
$$\alpha \cdot (A-1) \cdot b \leq \left( \left(\frac{1}{2} + \frac{(1-\alpha)A + \alpha}{B} \right) \cdot A + \frac{B}{2} \right) \cdot b .$$
Solving for $\alpha$, we get that 
$$\alpha \leq \frac{2A^2+A B+B^2}{2(A-1) (A+B)}.$$
Recall we have the freedom to set $A > B > 1$ as we wish.
By choosing $B = (\sqrt{2} - 1)A$, we get:
$$\alpha \leq \frac{(2\sqrt{2} -1) A}{2 (A-1)}.$$

We can then take $A \rightarrow \infty$ and arrive at $\alpha \leq \sqrt{2} - \frac{1}{2}$. 

\end{proof}

Furthermore, for the case of two bidders, we can show a useful upper and lower bound on how much the function should ``favor'' the higher bidder:
\begin{lemma}
\label{lem:two_bidder_bounds}
For the two bidder case, with $v_1 > v_2$, any  DSIC+MMIC+OCA-proof mechanism $a, p, \beta$ must satisfy:
\begin{equation}
\label{eq:lower_bound_two_bidder}
\frac{u_1(v1; v1) - v_2}{v_1 - v_2} \leq a(v_1, v_2) \leq \frac{3}{2} - \frac{u_1(v1;v1)}{v_2}.\end{equation}
Furthermore, if $u_1(v1;v1) > \frac{v_1+v_2}{2}$, then:
\begin{equation} a(v_1, v_2) \leq 1 - \frac{1}{v_2} \cdot \left(2u_1(v1;v1) - \frac{3}{2}v_1 +  (v_1 - u_1(v1;v1)) \log\frac{2(v_1 - u_1(v1;v1))}{v_1 - v_2} \right).\end{equation}
\end{lemma}


\begin{proof}
The lower bound stems from applying the OCA condition of \cref{lem:oca_joint} for omitting the lower bid $v_2$, together with the feasibility constraint of \cref{eq:feasibility}:

\begin{equation}
\label{eq:lower_bound_two_bidder_proof}
\begin{split}
& \alloc(v_1, v_2) \cdot v_1 + (1 - a(v_1, v_2)) \cdot v_2 \stackrel{\text{\cref{eq:feasibility}}}{\geq}  a(v_1, v_2) \cdot v_1 + a(v_2, v_1) \cdot v_2 \geq \\
& \int_0^{v_1} a(t, v_2) dt + a(v_2, v_1) \cdot v_2 \stackrel{\text{\cref{lem:two_bidders_cond}}}{\geq} \int_0^{v_1} a(t)dt = u_1(v_1 ; v_1),
\end{split}
\end{equation}
and then rearranging terms. 
The general upper bound follows using the OCA condition for omitting one of the bids if both bids are $v_1, v_1$:
\[
\begin{split}
& a(v_2, v_1) \cdot v_2 + \frac{1}{2}v_2 =
\frac{1}{2} \cdot v_1 + a(v_2, v_1) \cdot v_2 + \int_{v_2}^{v_1} \frac{1}{2} dt \stackrel{\text{Claim~\ref{clm:low_value_feasibility}, Integral Arithmetics}}{\geq} \\
& \alloc(v_1, v_1) \cdot v_1 + \left( \int_0^{v_2} a(t, v_1) dt + \int_{v_2}^{v_1} a(t, v_1) dt \right) = \alloc(v_1, v_1) \cdot v_1 + \int_0^{v_1} a(t, v_1) dt \stackrel{\text{\cref{lem:two_bidders_cond}}}\geq u_1(v1;v1),
\end{split}
\]
and then rearranging terms. 

For the case where $u_1(v1;v1) > \frac{v_1+v_2}{2}$, we can achieve a finer upper bound by splitting the integration between $v_2$ and $v_1$ into two segments: The first between $v_2$ and $2u_1(v1;v1)  - v_1$, and the second between $2u_1(v1;v1)  - v_1$ and $v_1$. For the second segment, we use the same $\frac{1}{2}$ upper bound that comes from Claim~\ref{clm:low_value_feasibility}
 as before. For the first segment, we can utilize the lower bound of \cref{eq:lower_bound_two_bidder} we just proved. We thus write:
\[
\begin{split}
&\frac{3}{2} \cdot v_1 + (1 - a(v_1, v_2)) \cdot v_2 + (v_1 - u_1(v1; v1))\int_{v_2}^{2u_1(v1;v1)  - v_1}  \frac{1}{v_1 - t} dt  - u_1(v1;v1) \stackrel{\text{Arithmetics}}{\geq} \\
& \frac{1}{2} \cdot v_1 + (1 - a(v_1, v_2)) \cdot v_2 + \int_{v_2}^{2u_1(v1;v1)  - v_1} (1 - \frac{u_1(v1;v1) - t}{v_1 - t}) dt + (v_1 - u_1(v1;v1)) \stackrel{\text{\cref{eq:lower_bound_two_bidder}}}{\geq} \\
& \frac{1}{2} \cdot v_1 + a(v_2, v_1) \cdot v_2 + \int_{v_2}^{2u_1(v1;v1)  - v_1}  (1 - a(v_1, t)) dt + \int_{2u_1(v1;v1)  - v_1}^{v_1} \frac{1}{2} dt \stackrel{\text{\cref{eq:feasibility},Claim~\ref{clm:low_value_feasibility}, Integral Arithmetics}}{\geq} \\
& \alloc(v_1, v_1) \cdot v_1 + \left( \int_0^{v_2} a(t, v_1) dt + \int_{v_2}^{v_1} a(t, v_1) dt \right) = \alloc(v_1, v_1) \cdot v_1 + \int_0^{v_1} a(t, v_1) dt \stackrel{\text{\cref{lem:two_bidders_cond}}}\geq u_1(v1;v1).
\end{split}
\]

Then, by rearranging terms:
\[
\begin{split}
    & a(v_1, v_2) \leq 1 - \frac{1}{v_2} \cdot \left(2u_1(v1;v1) - \frac{3}{2}v_1 - (v_1 - u_1(v1;v1))\int_{v_2}^{2u_1(v1;v1)  - v_1} \frac{1}{v_1 - t} dt \right) = \\
    & 1 - \frac{1}{v_2} \cdot \left(2u_1(v1;v1) - \frac{3}{2}v_1 +  (v_1 - u_1(v1;v1)) \log\frac{2(v_1 - u_1(v1;v1))}{v_1 - v_2} \right)
    \end{split}
\]

\end{proof}

\begin{corollary}
\label{corr:efficiency_bound}
For any DSIC+MMIC+OCA-proof anonymous randomized mechanism, there is a value $v_1$ such that $u_1(v_1;v_1) \leq 0.842 \cdot v_1$.
\end{corollary}

\begin{proof}
Restricting attention to values of $v_1, v_2$ such that $2u_1(v1;v1) - v_1 > v_2$, we have, by rearranging the terms:
$v_1 - v_2 > 2v_1 - 2u_1(v1;v1)$,
and so:
$\log \frac{2(v_1 - u_1(v1;v1))}{v_1 - v_2} < \log 1 = 0$.
Thus, as a function of $u_1(v1;v1)$, the extended upper bound of \cref{lem:two_bidder_bounds} is monotonically decreasing, while the lower bound is monotonically increasing with $u_1(v1;v1)$. Thus, if for some $v_1, v_2$ we find some value of $u_1(v1;v1)$ so that the lower bound \textit{exceeds} the upper bound, it will hold for any higher values of $u_1(v1;v1)$ as well. Given such ``witness'', we could thus conclude that $u_1(v1;v1)$ must be lower than this value. For the choice of $v_1 = 19.8, v_2 = 2.4$, this is satisfied for $u_1(v1;v1) = 0.842 \cdot v_1$. 
\end{proof}

The implication of the above lemma is that since in the single bidder case the utility of the single bidder is exactly the joint utility (as all payments are burned by \cref{lem:SingleBidderRevenue}), and the optimal joint utility is $v_1$, the efficiency ratio is thus bounded by $0.842$ in the worst case. However, notice that despite some similarity to the statement of \cref{thm:allocation_bound}, in the sense that the lemma does not generally restrict the allocation probability in the single bidder case, i.e., it is possible that the allocation function allocates with higher probability than $0.842$. \cref{thm:allocation_bound} shows a bound on the allocation probability that is true for \textit{any} value of $v_1$. 

\section{Discussion}

Our work provides a rich characterization of the space of DSIC, MMIC, and OCA-proof mechanisms and their intersections.
Our main result for deterministic mechanisms (\cref{thm:det_impossibility}), shows impossibility of a non-trivial deterministic DSIC, MMIC and $1$-OCA-proof mechanism.
Although the proof we give goes through characterizing the spaces of DSIC and $1$-OCA-proof mechanisms, versus MMIC and $1$-OCA-proof mechanisms (in \cref{thm:dsic_1oca} and \cref{thm:mmic_1oca}), there are other ways to show it using the building blocks we provide. One way is by using our result for randomized mechanisms in \cref{thm:allocation_bound}.
There, we show that we cannot allocate the item in the single-bidder case with probability higher than $\sqrt{2} - \frac{1}{2} < 1$. 
Together with \cref{lem:always_or_never_allocated}, this shows the impossibility. 

Our characterization of DSIC+$1$-OCA-proof and MMIC+$1$-OCA-proof mechanisms extends known incidental results in the literature. For example, Table~1 in \cite{roughgarden2021transaction} notes that the first-price auction is MMIC+$1$-OCA-proof, the second-price auction is DSIC+$1$-OCA-proof, and that EIP-1559 (which is essentially a first-price auction with a dynamically adjusted burned reserve price) is MMIC+$1$-OCA-proof. Our characterization confirms these results, but also precisely extends them: In essence, it shows that all MMIC+$1$-OCA-proof are ``EIP-1559 - type'' mechanisms, in the sense of having a threshold bid $r$ and a constant burn, albeit with the freedom to choose any monotone payment function and not necessarily first-price payments. Similarly, for DSIC+$1$-OCA-proof mechanisms, any second-price auction with a constant burn is possible.
These results are made more important by our impossibility result (\cref{thm:det_impossibility}), since if all requirements cannot be satisfied at once, satisfying at least two of the three may be a good compromise.

Throughout the work, we analyze anonymous mechanisms.
While this is a common assumption in the literature, we wish to explore how much it affects the impossibility result of \cref{thm:det_impossibility}.
In Appendix~\ref{app:nonAnon}, we remove this restriction and develop a characterization of all, not necessarily anonymous, DSIC+MMIC+$1$-OCA-proof mechanisms.
We find that non-anonymous mechanisms that satisfy these desired properties are very restricted: all of them must have some constant unique bidder (among all possible named bidders) that can be allocated the item, where the auction follows the form of a burned posted-price auction.
This limitation is reminiscent of the ``2-user-friendly'' impossibility in Theorem 6.1 of \cite{chung2023foundations}.


\printbibliography

\appendix

\section{Missing Proofs}
\label[appendix]{app:missing_proofs}

\SimpleJointBound*
\begin{proof}
\[
\begin{split}
& u_{joint}(\mathbf{b}) = \sum_{i=1}^n \left( v_i \alloc(b_i, \mathbf{b}_{-i}) - \burn(b_i, \mathbf{b}_{-i})\right) - \sum_{i=n+1}^{n'} \burn(b_i, \mathbf{b}_{-i}) 
\stackrel{\text{\cref{eq:BB}}}{\leq} \\
& \sum_{i=1}^n v_i \alloc(b_i, \mathbf{b}_{-i}) \leq \max \mathbf{v} \sum_{i=1}^n \alloc(b_i, \mathbf{b}_{-i}) \stackrel{\text{\cref{eq:feasibility}}}{\leq} \max \mathbf{v}.
\end{split}
\]
\end{proof}

\CoalitionJointBound*

\begin{proof}
The expression for the joint utility under the off-chain agreement $\Omega_{\mathbf{v}}$ takes a similar form to that of the utility of the miner and the $c$ colluding bidders, but also includes the utility of all other bidders.
By the individual rationality condition of \cref{eq:IR_util}, it therefore must be at least as large. 

\[
\begin{split}
u_{joint}(\Omega_{\mathbf{v}} ; \mathbf{v})
&
\stackrel{\text{\cref{eq:joint_util}}}{=}
u_{miner}(\Omega_{\mathbf{v}} ; \mathbf{v}) + \sum_{i=1}^n u_i(\Omega_{\mathbf{v}_i}, \Omega_{\mathbf{v}_{-i}} ; v_i)
\\&
=
u_{miner}(\Omega_{\mathbf{v}} ; \mathbf{v}) + \sum_{i \in C} u_i(\Omega_{\mathbf{v}_i}, \Omega_{\mathbf{v}_{-i}} ; v_i) + \sum_{i \not \in C} u_i(\Omega_{\mathbf{v}_i}, \Omega_{\mathbf{v}_{-i}} ; v_i)
\\&
\stackrel{\text{\cref{eq:IR_util}}}{\geq}
u_{miner}(\Omega_{\mathbf{v}} ; \mathbf{v}) + \sum_{i \in C} u_i(\Omega_{\mathbf{v}_i}, \Omega_{\mathbf{v}_{-i}} ; v_i).
\end{split}
\]
\end{proof}

\OCAJoint*
\begin{proof}
(OCA-proofness $\implies$ \cref{eq:general_oca})

If \gls{OCA}-proofness holds for any coalition $C$, then it also hold for the grand coalition of all $n$ users.
For any $\mathbf{v}$ and corresponding $\Omega_{\mathbf{v}}$ choose $C = [n]$, then the condition of \cref{eq:oca_cond} takes the form:

$$u_{joint}(\mathbf{v} ; \mathbf{v}) \geq u_{miner}(\Omega_{\mathbf{v}} ; \mathbf{v}) + \sum_{i \in C} u_i(\Omega_{\mathbf{v}_i}, \Omega_{\mathbf{v}_{-i}} ; v_i) = u_{miner}(\Omega_{\mathbf{v}} ; \mathbf{v}) + \sum_{i=1}^n u_i(\Omega_{\mathbf{v}_i}, \Omega_{\mathbf{v}_{-i} } ; v_i)\stackrel{\text{\cref{eq:joint_util}}}{=} u_{joint}(\Omega_{\mathbf{v}} ; \mathbf{v}).$$

(\cref{eq:general_oca} $\implies$ OCA-proofness)

This follows as an immediate corollary of Claim~\ref{clm:coalition_joint_bound}: For any $\mathbf{v}$, corresponding $\Omega_{\mathbf{v}}$ and coalition $C$, it then holds that:

$$u_{joint}(\mathbf{v} ; \mathbf{v}) \geq u_{joint}(\Omega_{\mathbf{v}} ; \mathbf{v}) \stackrel{\text{Claim~\ref{clm:coalition_joint_bound}}}{\geq} u_{miner}(\Omega_{\mathbf{v}} ; \mathbf{v}) + \sum_{i \in C} u_i(\Omega_{\mathbf{v}_i}, \Omega_{\mathbf{v}_{-i}} ; v_i),$$

but this is exactly the condition of \cref{eq:oca_cond}.

\SingleBidderRevenue*
\begin{proof}
By Myerson's lemma (given in~\cref{fct:MyersonLemma}), the allocation rule is monotone and the payment is uniquely determined by it. 

In the single bidder case, for a bidder with value $v_1$ and bid $b_1$ (i.e., $\mathbf{v} = (v_1)$, with an off-chain agreement of the form $\Omega_{\mathbf{v}} = (b_1)$, no fake bids or omitting by the miner, and the $1$-bidder coalition $C = \{1\}$), the OCA-proofness condition of \cref{eq:oca_cond} takes the form:

    $$v_1 \cdot \alloc(v_1) - \burn(v_1) = u_{joint}(v_1 ;  v_1) \geq u_{miner}(b_1 ; v_1) + u_1(b_1 ; v_1) = v_1 \cdot \alloc(b_1) - \burn(b_1).$$

Notice that this is exactly the incentive-compatibility condition for the single bidder, if we replace payment with burn.
Thus, we can follow Myerson's arguments to determine that the burn rule is uniquely determined by the allocation due to the OCA-proofness condition, in exactly the same way as the payment rule is determined due to DSIC. 

\end{proof}

\end{proof}

\DSICOCA*
\begin{proof}
($\implies$ All DSIC+$1$-OCA-proof mechanisms must be of the form specified)

Given the characterization of \cref{lem:general_oca_char}, the only allocation and burn rules possible are as described. However, since we assume the DSIC property, the payment is uniquely determined by the allocation, and since the allocations all fall within the class of second-price auctions with reserve $r$ allocations, then so do the corresponding payments. 

($\implies$ All $2$-nd price auctions with reserve $r$ and constant burn $r$ are DSIC+$1$-OCA-proof)

We know that the $2$-nd price auction with reserve $r$ is \gls{DSIC} \cite{myerson1981optimal}, and we furthermore know that the $1$-OCA-proofness property is satisfied due to the characterization of \cref{lem:general_oca_char}.

\end{proof}

\MMICOCA*
\begin{proof}
($\implies$ All MMIC+$1$-OCA-proof mechanisms must be of the form specified)

Given the characterization of \cref{lem:general_oca_char}, the only allocation and burn rules possible are as described. Regarding payment, consider bidder $i$ with valuation $v_i$, and two sets of bids of other bidders, $\mathbf{v_{-i}}, \mathbf{b_{-i}}$ so that $\tilde{\alloc}(v_i, \mathbf{v}_{-i}) = \tilde{\alloc}(v_i, \mathbf{b}_{-i}) = i$. Then, it must hold that $\pay(v_i, \mathbf{v}_{-i}) - \burn(v_i, \mathbf{v}_{-i}) = \pay(v_i, \mathbf{b}_{-i}) - \burn(v_i, \mathbf{b}_{-i}) $, or otherwise the miner would omit all bids $\mathbf{v}_{-i}$ and add fake bids $\mathbf{b}_{-i}$ (or vice versa). Since the burn is a constant $r$, this implies that $\pay(v_i, \mathbf{v}_{-i}) = \pay(v_i, \mathbf{b}_{-i})$. I.e., the payment is uniquely determined by the winner's bid $v_i$ (and must be higher than the burn $r$ by \cref{eq:BB}). We can thus write $\tilde{\pay}(\mathbf{b}) = f(\max \mathbf{b})$ for some function $f(v) \geq r$. 

We now show that $f(v)$ must be monotone whenever $v \geq r$. If not, then there are such $v' > v \geq r$ so that $f(v') < f(v)$. Consider two bidders who bid $v', v$. Without miner intervention, the miner would receive payment $f(v') - r$. If the miner omits the bid $v'$, it receives payment $f(v) - r > f(v') - r$, and so this contradicts MMIC.

($\implies$ All auctions of the form specified are MMIC+$1$-OCA-proof)

We have $1$-OCA-proofness by the characterization of \cref{lem:general_oca_char}. Regarding MMIC, the miner cannot gain by adding fake bids since the burn is constant and the payment only depends on the winning bid: If the added fake bids do not change the winning bid, then the miner revenue does not change, and if they do, since the item is allocated to the highest bidder, this means that one of the fake bids wins, and the miner revenue is then non-positive. The miner cannot gain by omitting bids since the burn is constant and the payment is monotone in the highest bid, and the highest bid can only decrease by omitting bids. 

\end{proof}

\homogeneousPayments*

\begin{proof}
For a DSIC mechanism we can write, by Fact~\ref{fct:MyersonLemma}:
\[
\begin{split}
p(\alpha b_i, \alpha \mathbf{b}_{-i})
&
=
a(\alpha b_i, \alpha \mathbf{b}_{-i}) \cdot \alpha b_i - \int_0^{\alpha b_i} a(t, \alpha \mathbf{b}_{-i}) dt
\\&
=
\alpha \cdot a(b_i, \mathbf{b}_{-i}) \cdot b_i - \alpha \int_0^{b_i} a(\alpha \cdot z, \alpha \mathbf{b}_{-i}) dz
\\&
=
\alpha \left( a(b_i, \mathbf{b}_{-i}) \cdot b_i - \int_0^{b_i} a(z, \mathbf{b}_{-i}) dz \right)
\\&
=
\alpha p(b_i, \mathbf{b}_{-i}).
\end{split}
\]
\end{proof}

\CTPA*
\begin{proof}
Let $\alpha$ be the constant allocation probability of the CTPA mechanism. Since in the single bidder case the allocation probability is $\alpha$ regardless of the bid, the burn must be $0$ as for any bid $b_1$:
$$\burn(b_1) \stackrel{\text{Claim~\ref{lem:SingleBidderRevenue}}}{=} \pay(b_1) \stackrel{\text{Fact.~\ref{fct:MyersonLemma}}}{=} \alloc(b_1) \cdot b_1 - \int_0^{b_1} \alloc(t) dt = \alpha \cdot b_1 - \int_0^{b_1} \alpha dt = 0.$$
Since the item's allocation probability is always $\alpha$, we can extend the argument of Claim~\ref{clm:simple_joint_bound} to show:
\begin{equation}
\label{eq:alpha_joint_bound}
\begin{split}& u_{joint}(\mathbf{b} ; \mathbf{v}) = \sum_{i=1}^n \left( v_i \alloc(b_i, \mathbf{b}_{-i}) - \burn(b_i, \mathbf{b}_{-i})\right) - \sum_{i=n+1}^{n'} \burn(b_i, \mathbf{b}_{-i}) 
\stackrel{\text{\cref{eq:BB}}}{\leq} \\
& \sum_{i=1}^n v_i \alloc(b_i, \mathbf{b}_{-i}) \leq \max \mathbf{v} \sum_{i=1}^n \alloc(b_i, \mathbf{b}_{-i}) \leq \alpha \cdot \max \mathbf{v} .
\end{split}
\end{equation}
Thus, for any vector of valuations $\mathbf{v}$, 
$$u_{joint}(\mathbf{v} ; \mathbf{v}) \stackrel{\text{\cref{eq:alpha_joint_bound}}}{\leq} \alpha \cdot \max \mathbf{v} = u_{joint}(\max \mathbf{v} ; \mathbf{v}). $$
If the inequality is strict, this violates the OCA-proofness condition of \cref{lem:oca_joint}.
Thus, we must always have $$u_{joint}(\mathbf{v} ; \mathbf{v}) = \alpha \max \mathbf{v},$$
i.e., the item must be allocated with the entire constant probability $\alpha$ to the highest bidder. However, this is a form of a second-price auction (but with limited allocation probability), and because of DSIC this means the payment must increase with the second highest bid, which contradicts MMIC. 
\end{proof}

\section{Non-anonymous Deterministic Mechanisms}
\label{app:nonAnon}

We quickly revisit our auction model to allow for non-anonymous mechanisms.
We now have a fixed set of $N$ possible named bidders, out of which some subset $I\subseteq N$ actually participates in the auction. The allocation, payment, and burn rule $\alloc,\pay,\burn$ now operate on the domain $R^I$ and may incorporate the identities of the bidders into the outcome. The DSIC, MMIC, OCA-proofness, and SCP notions remain the same, with the important emphasis that we still allow for the miner to issue fake bids. This is justified by the fact that an agent may hold multiple bidder identities, and so in the analysis perspective it is possible that the miner is in control of virtually any of the bidder identities. 

Importantly, \cref{lem:SingleBidderRevenue} ($0$ revenue in the single bidder case), Fact~\ref{fct:MyersonLemma} (Myerson's Lemma), \cref{thm:zero_revenue} ($0$ revenue in the general case) transfer to the non-anonymous case and we subsequently use them. 

We now characterize the space of DSIC+MMIC+$1$-OCA-proof non-anonymous mechanisms. 
We do so by proving progressive refinements of the space of possible mechanisms. 

\begin{lemma}
\label{lem:non_anonymous}
All \gls{DSIC}+\gls{MMIC}+$1$-\gls{OCA}-proof mechanisms satisfy the following properties:

\begin{enumerate}[leftmargin=*]
\item Each bidder $i$ has a personalized burn rate $r_i$ so that whenever bidder $i$ wins, exactly $r_i$ must be paid and burned. 

\item There is at most one unique bidder $i^*$ with $r_{i^*} \neq \infty$. 

\item For any given set of bidders $I$, the item is allocated only if $i^* \in I$, and the allocation rule is of the form $\tilde{a}(\mathbf{b}) = \begin{cases} i^* & b_{i^*} \geq r_{i^*}, \\
\emptyset & Otherwise.\end{cases}$
\end{enumerate}
\end{lemma}

\begin{proof}
We give a proof sketch, as the proof resembles the one for the anonymous case.

We prove each property separately.
\begin{enumerate}[leftmargin=*]
\item By \cref{lem:SingleBidderRevenue} and Fact~\ref{fct:MyersonLemma}, each bidder $i \in N$ has some $r_i$ so that in case $i$ is the single bidder, $\alloc(v_i) = \mathds{1}_{v_i \geq r_i}, \tilde{\pay}(v_i) = \tilde{\burn}(v_i) = r_i$. Thus, it must hold that whenever $i\in I$ and for any $v_i$ so that $\tilde{\alloc}(v_i, \mathbf{v}_{-i}) = i$, it also holds that $\tilde{\burn}(v_i, \mathbf{v}_{-i}) = r_i$. Since there is $0$ revenue, it also holds that $\tilde{\pay}(v_i, \mathbf{v}_{-i}) = r_i$.

\item Assume towards contradiction that there are two different bidders $i,j$ with $r_i \neq \infty, r_j \neq \infty$. Consider $I = \{i, j\}$. For $\mathbf{v}$ so that $v_i = r_i + 2, v_j = r_j + 1$, it must hold that $\tilde{\alloc}(\mathbf{v}) = i$, otherwise there is a $1$-OCA where the miner omits $v_j$. But since we already concluded that bidder $i$ pays $r_i$ in this case, by the critical bid property, for $\mathbf{v}'$ where $v'_i = r_i, v'_j = r_j + 1$, it also holds that $\tilde{\alloc}(\mathbf{v}) = i$. But this means there is a $1$-OCA where the miner omits $v'_i$. 

\item No other bidder besides $i^*$ can be allocated because all have infinite burn rates if allocated. Thus, if $i^* \not \in I$, then no bidder is ever allocated. If $i^* \in I$, and $v_{i^*} > r_{i^*}$, assume towards contradiction that $\tilde{\alloc}(\mathbf{v}) \neq i^*$, then it must be $\tilde{\alloc}(\mathbf{v}) = \emptyset$. But then there is a $1$-OCA where the miner omits all other bids and the joint utility is $v_{i^*} - r_{i^*} > 0$. If  $v_{i^*} < r_{i^*}$, assume towards contradiction that $\tilde{\alloc}(\mathbf{v}) = i^*$. Since the burn must be at most $v_{i^*}$, there is a $1$-OCA from the single $i^*$ bidder case where the miner adds the fake bids $\mathbf{v}_{-i}$ to reduce the burn. 
\end{enumerate}

\end{proof}

It is then straightforward to show that the class of mechanisms we identified cannot be further refined, and in fact satisfies a stronger property:

\begin{lemma}
The class of mechanisms characterized in \cref{lem:non_anonymous} is \gls{DSIC}+MMIC+SCP. 
\end{lemma}
\begin{proof}
    Consider a mechanism contained in the class.
    It is \gls{DSIC}, as per the definition of the class in \cref{lem:non_anonymous}, the mechanism satisfies the monotone allocation and critical bid payment conditions of Myerson's Lemma~\cref{fct:MyersonLemma}.
    The MMIC property is satisfied as the payment and burn rules are equal, implying that the miner always receives $0$ revenue, thereby preventing the miner from being able to unilaterally increase its revenue without potentially colluding with users.

    Finally, note that no beneficial collusion exists: the individual burn rate is fixed for each user, and cannot change as dependent on the bids of the other transactions included in the block, or the identities of the transactions' creators.

SCP is satisfied since if the coalition $C$ does not contain $i^*$, the aggregate utility of the coalition is $0$ with truthful reports, while non-positive under any manipulation. If the coalition does contain $i^*$, the SCP condition becomes the DSIC condition for agent $i^*$, which is satisfied as the mechanism is DSIC. 
    
\end{proof}

\end{document}

%% file: preamble.tex
\newif\iflong   
\newif\ifmorecites
\newif\ifanonymous
\newif\ifcomments   

\commentsfalse
\anonymousfalse

\providecommand{\keywords}[1]{\textbf{Keywords:} #1.}

\usepackage[letterpaper,margin=1in]{geometry}
\usepackage{comment}
\usepackage[utf8]{inputenc}

\usepackage{newunicodechar}
\newunicodechar{ϕ}{\ensuremath{\phi}}
\usepackage[backend=biber,style=alphabetic,maxbibnames=10,maxalphanames=10,minalphanames=10]{biblatex}
\usepackage[T1]{fontenc}    
\usepackage[english]{babel} 
\usepackage{csquotes}       
\usepackage{amsmath}        
\usepackage{dsfont}         
\usepackage{upgreek}        
\usepackage{mathtools}      
\usepackage{enumitem}       
\usepackage{graphicx}       
\graphicspath{ {./images/} }
\usepackage{booktabs}       
\usepackage{array,multirow} 
\usepackage{tabularx}       
\usepackage{amsthm}         
\usepackage{thmtools}
\usepackage{thm-restate}   
\usepackage{amsfonts}       
\usepackage{hyperref} 
\usepackage{xcolor}         
\hypersetup{                
    colorlinks,
    linkcolor={red!50!black},
    citecolor={blue!50!black},
    urlcolor={blue!80!black}
}
\usepackage{etoolbox}\appto\UrlBreaks{\do\-}
\usepackage[capitalise]{cleveref}       

\usepackage{xparse}

\ifcomments\fi  
\usepackage[colorinlistoftodos,prependcaption,textsize=tiny,textwidth=\marginparwidth]{todonotes}

\newtheorem{definition}{Definition}[section]
\newtheorem{theorem}[definition]{Theorem}
\newtheorem{claim}[definition]{Claim}
\newtheorem{fact}[definition]{Fact}
\newtheorem{lemma}[definition]{Lemma}
\newtheorem{corollary}[definition]{Corollary}

\newtheorem{remark}[definition]{Remark}
\newtheorem{example}[definition]{Example}

\crefname{definition}{Def.}{Defs.}
\crefname{theorem}{Theorem}{Theorems}
\crefname{claim}{Claim}{Claims}
\crefname{fact}{Fact}{Fact}
\crefname{lemma}{Lemma}{Lemmas}
\crefname{corollary}{Corollary}{Corollaries}
\crefname{example}{Example}{Examples}
\crefname{remark}{Remark}{Remarks}

\newcommand{\define}{\stackrel{\mathclap{\tiny\mbox{def}}}{=}}

\usepackage[symbols,acronym,nonumberlist,nogroupskip,section=subsection,numberedsection,stylemods={mcols,longbooktabs}]{glossaries-extra}
\makenoidxglossaries

\glssetcategoryattribute{acronym}{nohyper}{true} 
\setabbreviationstyle[acronym]{long-short}  
\newacronym{EVM}{EVM}{Ethereum virtual machine}
\newacronym{DeFi}{DeFi}{decentralized finance}
\newacronym{PoW}{PoW}{Proof-of-Work}
\newacronym{PoS}{PoS}{Proof-of-Stake}
\newacronym{MEV}{MEV}{miner-extractable value}
\newacronym{block-DAG}{block-DAGs}{block directed-acyclic-graph}
\newacronym{DAA}{DAA}{difficulty-adjustment algorithm}
\newacronym{MDP}{MDP}{Markov decision-process}
\newacronym{DQL}{DQL}{Deep-Q-learning}
\newacronym{RL}{RL}{reinforcement learning}
\newacronym{ML}{ML}{machine learning}
\newacronym{AI}{AIs}{artificial intelligence}
\newacronym{PDF}{PDF}{probability density function}
\newacronym{CDF}{CDF}{cumulative density function}
\newacronym{AMM}{AMM}{automated market maker}
\newacronym{USD}{USD}{United States Dollar}
\newacronym{IP}{IP}{Internet Protocol}
\newacronym{LP}{LP}{Liquidity Provider}
\newacronym{LT}{LT}{Liquidity Taker}
\newacronym{APY}{APY}{annual percentage yield}
\newacronym{PID}{PID}{proportional integral derivative}
\newacronym{UTXO}{UTXO}{unspent transaction output}
\newacronym{YAML}{YAML}{YAML Ain't Markup Language}
\newacronym{TD}{TD}{total difficulty}
\newacronym{geth}{geth}{Go Ethereum}
\newacronym{WETH}{WETH}{Wrapped Ethereum}
\newacronym{ASIC}{ASIC}{Application Specific Integrated Circuit}
\newacronym{RPC}{RPC}{remote procedure call}
\newacronym{RUM}{RUM}{riskless uncle maker}
\newacronym{PUM}{PUM}{preemptive uncle maker}
\newacronym{URL}{URL}{uniform resource locator}
\newacronym{SSD}{SSD}{solid state drive}
\newacronym{EIP}{EIP}{Ethereum improvement proposal}
\newacronym{CPU}{CPU}{central processing unit}
\newacronym{RAM}{RAM}{random-access memory}
\newacronym{mempool}{mempool}{memory pool}
\newacronym{TFM}{TFM}{transaction fee mechanism}
\newacronym{iid}{i.i.d.}{independent and identically distributed}
\newacronym{wrt}{w.r.t.}{with regard to}
\newacronym{wlog}{w.l.o.g.}{without loss of generality}
\newacronym{TTL}{TTL}{time to live}
\newacronym{DSIC}{DSIC}{dominant strategy incentive-compatible}
\newacronym{EPIC}{EPIC}{ex-post incentive-compatible}
\newacronym{BIC}{BIC}{Bayesian incentive-compatible}
\newacronym{MMIC}{MMIC}{myopic miner incentive-compatible}
\newacronym{FMIC}{FMIC}{farsighted miner incentive-compatible}
\newacronym{OCA}{OCA}{off-chain-agreement}
\newacronym{SCP}{SCP}{side-contract-proof}
\newacronym{PABGA}{PABGA}{pay-as-bid greedy auction}
\newacronym{SPA}{SPA}{second price auction}
\newacronym{UPGA}{UPGA}{uniform-price greedy auction}
\newacronym{BNE}{BNE}{Bayesian-Nash equilibrium}
\newacronym{EPIR}{EPIR}{ex-post individually rational}
\newacronym{EPBB}{EPBB}{ex-post burn balanced}
\newacronym{GTA}{GTA}{good toy auction}
\newacronym{GTFBA}{GTFBA}{good toy finite-blocksize auction}
\newacronym{QoS}{QoS}{quality of service}
\newacronym{CTPA}{CTPA}{constant total probability of allocation}

\glsxtrnewsymbol[description={
    An auction.
}]{auction}{
    \ensuremath{A}
}
\newcommand{\auction}{{\gls[hyper=false]{auction}}}

\glsxtrnewsymbol[description={
    A transaction.
}]{tx}{
    \ensuremath{tx}
}

\glsxtrnewsymbol[description={
    Payment rule.
}]{pay}{
    \ensuremath{p}
}
\newcommand{\pay}{{\gls[hyper=false]{pay}}}

\glsxtrnewsymbol[description={
    Burning rule.
}]{burn}{
    \ensuremath{\beta}
}
\newcommand{\burn}{{\gls[hyper=false]{burn}}}

\glsxtrnewsymbol[description={
    Update rule.
}]{update}{
    \ensuremath{\upsilon}
}

\glsxtrnewsymbol[description={
    \gls[hyper=false]{TTL} of a transaction.
}]{ttl}{
    \ensuremath{t}
}

\glsxtrnewsymbol[description={
    Reserve price of an auction.
}]{reserve}{
    \ensuremath{r}
}
\newcommand{\reserve}{{\gls[hyper=false]{reserve}}}

\glsxtrnewsymbol[description={
    \gls[hyper=false]{TTL} of past transactions.  
}]{mempoolSymb}{
    \ensuremath{\mu}
}

\glsxtrnewsymbol[description={
    Immediacy ratio for our non-myopic allocation rule.
}]{imratio}{
    \ensuremath{\ell}
}

\glsxtrnewsymbol[description={
    Transaction schedule function.
}]{adversary}{
    \ensuremath{\psi}
}

\glsxtrnewsymbol[description={
    Allocation function.
}]{allocation}{
    \ensuremath{a}
}
\newcommand{\alloc}{{\gls[hyper=false]{allocation}}}

\glsxtrnewsymbol[description={
    Transaction fee of some transaction, in tokens.
}]{fee}{
    \ensuremath{b}
}
\newcommand{\fee}{{\gls[hyper=false]{fee}}}

\glsxtrnewsymbol[description={
    Maximal \gls[hyper=false]{TTL} of transactions in the blockchain, in blocks.
}]{maxttl}{
    \ensuremath{d}
}

\glsxtrnewsymbol[description={
    Predefined maximal block-size, in bytes.
}]{blocksize}{
    \ensuremath{\mathcal{B}}
}

\glsxtrnewsymbol[description={
    Miner's relative mining power, as a fraction.
}]{minerRatio}{
    \ensuremath{\alpha}
}

\glsxtrnewsymbol[description={
    Miner discount factor.
}]{discount}{
    \ensuremath{\lambda}
}

\glsxtrnewsymbol[description={
    Miner planning horizon.
}]{horizon}{
    \ensuremath{T}
}

\glsxtrnewsymbol[description={
    Miner revenue.
}]{revenue}{
    \ensuremath{u}
}

\glsxtrnewsymbol[description={
    Number of all bids with $b_i \geq \reserve$.
}]{bidsMoreReserve}{
    \ensuremath{S_{\geq \reserve}}
}

\NewDocumentCommand{\expect}{ e{_} s o >{\SplitArgument{1}{|}}m }{%
  \operatorname{E}
  \IfValueT{#1}{{\!}_{#1}}
  \IfBooleanTF{#2}{
    \expectarg*{\expectvar#4}%
  }{
    \IfNoValueTF{#3}{
      \expectarg{\expectvar#4}%
    }{
      \expectarg[#3]{\expectvar#4}%
    }%
  }%
}
\NewDocumentCommand{\expectvar}{mm}{%
  #1\IfValueT{#2}{\nonscript\;\delimsize\vert\nonscript\;#2}%
}
\DeclarePairedDelimiterX{\expectarg}[1]{[}{]}{#1}